\documentclass[aps,prx,twocolumn]{revtex4-1}

\usepackage{mathrsfs}
\usepackage{bbm}
\usepackage{amsmath}
\usepackage{amssymb}
\usepackage{amsthm}
\usepackage{graphicx}
\usepackage{MnSymbol}
\usepackage{mathtools}
\usepackage{stmaryrd}
\usepackage{bbold}
\usepackage[italicdiff]{physics}
\usepackage{dsfont}

\makeatletter

\usepackage{hyperref}
\usepackage{color}
\definecolor{myblue}{RGB}{0,50,200}
\hypersetup{
colorlinks,
citecolor=myblue,
linkcolor=myblue,
urlcolor=myblue
}
\allowdisplaybreaks

\newtheorem{theorem}{Theorem}

\newtheorem{lemma}[theorem]{Lemma}

\newcommand{\mca}{\mathcal}
\newcommand{\mbb}{\mathbb}
\newcommand{\mds}{\mathds}

\newcommand{\kvec}[1]{|#1)}
\newcommand{\bvec}[1]{(#1|}

\newcommand{\vpi}{\vb*{\pi}}

\DeclareMathOperator{\mvar}{var}
\DeclareMathOperator{\sign}{sgn}

\begin{document}
\title{Fundamental bounds on precision and response for quantum trajectory observables}

\author{Tan Van Vu}
\email{tan.vu@yukawa.kyoto-u.ac.jp}
\affiliation{Center for Gravitational Physics and Quantum Information, Yukawa Institute for Theoretical Physics, Kyoto University, Kitashirakawa Oiwakecho, Sakyo-ku, Kyoto 606-8502, Japan}

\date{\today}

\begin{abstract}
The precision and response of trajectory observables offer valuable insights into the behavior of nonequilibrium systems. For classical systems, trade-offs between these characteristics and thermodynamic costs, such as entropy production and dynamical activity, have been established through uncertainty relations. Quantum systems, however, present unique challenges, where quantum coherence can enhance precision and violate classical uncertainty relations. In this study, we derive trade-off relations for stochastic observables in Markovian open quantum systems. Specifically, we present three key results: (i) a quantum generalization of the thermo-kinetic uncertainty relation, which bounds the relative fluctuations of currents in terms of entropy production and dynamical activity; (ii) a quantum inverse uncertainty relation, which constrains the relative fluctuations of arbitrary counting observables based on their instantaneous fluctuations and the spectral gap of the symmetrized Liouvillian; and (iii) a quantum response kinetic uncertainty relation, which bounds the response of general observables to kinetic perturbations in terms of dynamical activity. These fundamental bounds, validated numerically using a three-level maser and a boundary-driven XXZ spin chain, provide a comprehensive framework for understanding the interplay between precision, response, and thermodynamic costs in quantum systems.
\end{abstract}

\pacs{}
\maketitle

\section{Introduction}
Physical systems in nature and experiments are generally driven out of equilibrium and subject to significant fluctuations. For Markovian systems, ensemble dynamics can be unraveled into stochastic trajectories, where physically relevant observables can be defined as stochastic quantities \cite{Landi.2024.PRXQ}. Examples include the particle or heat current transported from a source to a target, as well as the distance a molecular motor travels within a finite time. For such observables, two critical aspects are the extent of their relative fluctuations (i.e., precision) and their sensitivity to small perturbations in control parameters (i.e., response). Understanding how these factors are constrained by thermodynamic costs is not only of theoretical importance but also provides valuable tools for thermodynamic inference \cite{Seifert.2019.ARCMP}. Over the past few decades, this area has seen substantial progress, particularly with the development of quantum stochastic thermodynamics for microscopic systems \cite{Sekimoto.2010,Seifert.2012.RPP,Vinjanampathy.2016.CP,Goold.2016.JPA,Deffner.2019}.

{\renewcommand{\arraystretch}{2.5}
\begin{table*}[t]
\label{table:res.sum}
\begin{tabular}{|c|c|c|}
\hline
Main results & Formulation & Applicable observables\\[2.0ex]
\hline\hline
Quantum thermo-kinetic uncertainty relation [Eq.~\eqref{eq:main.result.1}] & $\dfrac{F_\phi}{(1+\delta_\phi)^2}\ge\dfrac{4a}{\sigma^2}\Phi\qty(\dfrac{\sigma}{2a})^2\ge\max\qty(\dfrac{2}{\sigma},\dfrac{1}{a})$ & currents \\ [2.0ex]
\hline
Quantum inverse uncertainty relation [Eq.~\eqref{eq:main.result.2}] & $F_\phi\le\dfrac{\ev{J_2,\pi}}{\ev{J_1,\pi}^2}\qty(1+\dfrac{2\kappa}{g_s})$ & counting observables\\[2.0ex]
\hline
Quantum response kinetic uncertainty relation [Eq.~\eqref{eq:main.result.3}] & $\dfrac{\|\nabla\ev{f(\vb*{\phi})}\|_1^2}{\mvar[f(\vb*{\phi})]}\le\tau a$ & arbitrary\\[2.0ex]
\hline
\end{tabular}
\centering
\caption{Summary of our main results. The first result applies specifically to current-type observables, while the second one is broadly applicable to any counting observables. Remarkably, the third result extends even further, accommodating any function $f(\vb*{\phi})$, where $\vb*{\phi}$ represents a vector of arbitrary counting observables. All the relations universally hold for arbitrary finite times.}
\end{table*}}

The trade-off between precision and thermodynamic costs has recently been explored through the lens of uncertainty relations, specifically the thermodynamic uncertainty relation (TUR) and the kinetic uncertainty relation (KUR). These relations assert that achieving high precision always incurs a cost. The TUR establishes that the precision of any time-integrated current cannot be enhanced without increasing entropy production \cite{Barato.2015.PRL,Gingrich.2016.PRL,Horowitz.2017.PRE,Hasegawa.2019.PRL,Timpanaro.2019.PRL,Hasegawa.2019.PRE,Dechant.2020.PNAS,Vo.2020.PRE,Ray.2023.PRE,Horowitz.2020.NP}. Mathematically, it is expressed as an inequality between current fluctuations and dissipation,
\begin{equation}\label{eq:org.TUR}
F_\phi\coloneqq\tau\frac{\mvar[\phi]}{\ev{\phi}^2}\ge\frac{2}{\sigma},
\end{equation}
where $\ev{\phi}$ and $\mvar[\phi]$ are the mean and variance of the current $\phi$ over the operational time $\tau$, and $\sigma$ is the irreversible entropy production rate. The TUR has significant applications across various fields, including heat engines \cite{Pietzonka.2018.PRL}, molecular motors \cite{Pietzonka.2016.JSM}, anomalous diffusion \cite{Hartich.2021.PRL}, and dissipation estimation \cite{Li.2019.NC,Manikandan.2020.PRL,Vu.2020.PRE,Otsubo.2020.PRE}.
The KUR, on the other hand, establishes a trade-off between the precision of arbitrary counting observables and dynamical activity \cite{Garrahan.2017.PRE,Terlizzi.2019.JPA,Prech.2024.arxiv.CUR,Macieszczak.2024.arxiv}, explicitly given by $F_\phi \ge 1/a$, where $a$ quantifies the jump frequency. Notably, these relations can be unified into a tighter bound known as the thermo-kinetic uncertainty relation (TKUR) \cite{Vo.2022.JPA}.
The TUR generally holds for steady-state systems described by classical Markov jump processes and overdamped Langevin dynamics. It has been extended to arbitrary initial states and time-dependent driving \cite{Dechant.2018.JSM,Liu.2020.PRL,Vu.2020.PRR,Koyuk.2020.PRL}. However, violations of the TUR have been observed in other dynamics, including underdamped \cite{Vu.2019.PRE.UnderdampedTUR,Lee.2019.PRE,Pietzonka.2022.PRL} and quantum regimes \cite{Agarwalla.2018.PRB,Ptaszynski.2018.PRB,Liu.2019.PRE,Saryal.2019.PRE,Cangemi.2020.PRB,Friedman.2020.PRB,Kalaee.2021.PRE,Menczel.2021.JPA,Bret.2021.PRE,Sacchi.2021.PRE,Lu.2022.PRB,Gerry.2022.PRB,Das.2023.PRE,Manzano.2023.PRR,Singh.2023.PRA,Farina.2024.PRE,Palmqvist.2024.arxiv}. In particular, it has been shown that quantum coherence can significantly enhance precision and play a crucial role in violating the TUR. Although several generalizations of these uncertainty relations to quantum domains have been proposed \cite{Guarnieri.2019.PRR,Carollo.2019.PRL,Hasegawa.2020.PRL,Miller.2021.PRL.TUR,Vu.2022.PRL.TUR,Hasegawa.2023.NC,Prech.2025.PRL}, identifying conditions under which quantum coherence leads to TKUR violations remains elusive. This underscores the need for a novel quantum bound that clarifies the role of quantum coherence to address this gap.

As a complementary aspect of precision, the response of observables to small perturbations is of significant interest, as it provides a deeper characterization of physical systems. For systems near equilibrium, response theory has been well-established through the fluctuation-dissipation theorem (FDT) \cite{Kubo.1991}. Numerous generalizations of the FDT to far-from-equilibrium scenarios have been proposed \cite{Agarwal.1972,Baiesi.2009.PRL,Seifert.2010.EPL,Marconi.2008.PR}, with a primary focus on understanding the violation of the FDT \cite{Harada.2005.PRL}.
In recent years, the theory of static response has unveiled a close relationship between the response of observables and thermodynamic costs \cite{Owen.2020.PRX,Martins.2023.PRE,Gao.2024.EPL,Aslyamov.2024.PRL,Ptaszynski.2024.PRL,Zheng.2024.arxiv,Aslyamov.2024.arxiv,Liu.2024.arxiv}. Specifically, it has been shown that the response of observables to kinetic perturbations is bounded above by their dynamical fluctuations and thermodynamic quantities such as entropy production and dynamical activity. While these findings have been extensively developed for classical systems in nonequilibrium steady states, a comprehensive theory for open quantum systems is still lacking.

In this paper, we advance the understanding of quantum trajectory observables by deriving fundamental bounds for their precision and response, building upon the aforementioned two backgrounds. Focusing on Markovian open quantum dynamics with quantum jump and diffusion unravelings, we present three main results (see Table \ref{table:res.sum} for summary). First, we derive a quantum generalization of the TKUR [cf.~Eq.~\eqref{eq:main.result.1}], which establishes a lower bound on the relative fluctuations of currents in terms of entropy production and dynamical activity. This relation explicitly highlights the role of quantum coherence in enhancing current precision and violating classical uncertainty relations. Applying this bound to quantum heat engines, we reveal a trade-off between power, efficiency, and fluctuation, offering novel insights into the design of heat engines capable of achieving the Carnot efficiency at finite power without divergent fluctuations. Next, we derive an upper bound on the relative fluctuation of arbitrary counting observables [cf.~Eq.~\eqref{eq:main.result.2}], referred to as the quantum inverse uncertainty relation. This bound shows that the relative fluctuation of observables is constrained by their instantaneous fluctuation and the spectral gap of the symmetrized Liouvillian. When combined with the quantum TKUR, it yields a ``sandwich'' bound on the relative fluctuations of currents. Finally, we derive a quantum response kinetic uncertainty relation [cf.~Eq.~\eqref{eq:main.result.3}], which provides an upper bound on the response of general observables to kinetic perturbations in terms of their fluctuations and dynamical activity. Notably, this relation is quantitatively tighter than the KUR and recovers the KUR in the classical limit. Our findings are validated numerically using a three-level maser engine and a quantum many-body spin system.

\section{Setup}
We consider a $d$-dimensional open quantum system, whose dynamics is described by the Gorini-Kossakowski-Sudarshan-Lindblad (GKSL) master equation \cite{Lindblad.1976.CMP,Gorini.1976.JMP}:
\begin{align}\label{eq:Lindblad.dyn}
\dot{\varrho}_t&=\mca{L}(\varrho_t),\\
\mca{L}(\circ)&\coloneq-i\comm{H}{\circ}+\sum_{k\ge 1}\qty(L_k\circ L_k^\dagger-\acomm{L_k^\dagger L_k}{\circ}/2).\notag
\end{align}
Here, $\varrho_t$ denotes the system's density matrix at time $t$, and both Hamiltonian $H$ and jump operators $\{L_k\}$ are time-independent.
We assume that after a sufficiently long time, the system relaxes toward a unique stationary state $\pi$, which can be nonequilibrium.
Throughout this study, both the Planck constant and Boltzmann constant are set to unity, $\hbar=k_B=1$.

We mainly focus on two scenarios that can be described by the equation \eqref{eq:Lindblad.dyn}.
The first is thermodynamically dissipative dynamics, wherein the system is attached to single or multiple heat baths.
Assuming that the coupling between the system and the baths is weak and the environment is memoryless, the time evolution of the system's state is governed by the GKSL equation.
Each jump operator $L_k$ can, for example, characterize a jump between the energy eigenstates.
To guarantee the thermodynamic consistency for thermodynamically dissipative dynamics, we assume the local detailed balance condition \cite{Horowitz.2013.NJP}, which is fulfilled in most cases of physical interest \cite{Manzano.2018.PRX}.
That is, each jump operator $L_k$ can be associated with a reversed jump $L_{k^*}$ such that $L_k=e^{\Delta s_k/2}L_{k^*}^\dagger$.
Here, $\Delta s_k$ denotes the entropy change of the environment due to the jump.
Note that $k^*=k$ is possible (i.e., $L_k$ is a self-adjoint operator and $\Delta s_k=0$).
The second scenario is generic Markovian dynamics, where we do not impose any conditions on jump operators, and they can have arbitrary forms.
Examples include quantum measurement processes, wherein $L_k$ represents a measurement operator performed on the system \cite{Wiseman.2009}.

We remark on the validity and applicability of the GKSL dynamics. Its derivation from microscopic principles relies on the Born-Markov approximation, which assumes weak coupling between the system and the environment, as well as a fast relaxation of the environment \cite{Breuer.2002}. Consequently, it fails to describe non-Markovian dynamics, where memory effects and finite environmental relaxation times play a crucial role. This limitation restricts its applicability to strongly interacting systems, where system-environment correlations and entanglement become significant. Despite these constraints, the GKSL dynamics remains a fundamental framework in nonequilibrium physics due to its mathematical tractability and its ability to capture key aspects of dissipative processes, quantum decoherence, and steady-state behavior \cite{Breuer.2002,Manzano.2022.QS}. It continues to serve as an indispensable tool for studying open quantum systems in controlled experimental and theoretical settings.

\subsection{Quantum jump unraveling and observables}
The dynamics of Markovian open quantum systems can be unraveled into quantum jump trajectories \cite{Horowitz.2012.PRE,Horowitz.2013.NJP,Manzano.2015.PRE,Miller.2021.PRE}.
That is, the GKSL dynamics \eqref{eq:Lindblad.dyn} can be interpreted as a stochastic process of the pure state $\ket{\psi_t}$ such that $\varrho_t=\mbb{E}[\dyad{\psi_t}]$, where the average $\mbb{E}[\cdot]$ is taken over all stochastic trajectories.
The method of quantum jumps is a convenient way to describe the evolution of a quantum system that is constantly monitored, particularly when it is coupled to dissipative environments or is subject to continuous measurements.
For a small time step $dt\ll 1$, the master equation $\varrho_{t+dt}=(\mbb{1}+\mca{L}dt)\varrho_t$ can be expressed in the Kraus representation $\varrho_{t+dt}=\sum_{k\ge 0}M_k\varrho_tM_k^\dagger$ with the operators given by
\begin{equation}
M_0\coloneq \mbb{1}-iH_{\rm eff}dt,~M_k\coloneq L_k\sqrt{dt}~(k\ge 1).
\end{equation}
Here, $H_{\rm eff}\coloneqq H-(i/2)\sum_{k\ge 1}L_k^\dagger L_k$ is the effective Hamiltonian and $\mbb{1}$ denotes the identity operator.
The operator $M_0$ represents a smooth nonunitary evolution, whereas operators $\{M_k\}_{k\ge 1}$ induce probabilistic jumps in the system state.
With this interpretation, the GKSL equation \eqref{eq:Lindblad.dyn} can be unraveled into stochastic trajectories, wherein the pure state $\ket{\psi_t}$ either smoothly evolves or discontinuously jumps to another pure state at random times.
Specifically, the time evolution of the pure state $\ket{\psi_t}$ can be described by the stochastic Schr{\"o}dinger equation \cite{Breuer.2002},
\begin{align}
	d\ket{\psi_t}&=\qty(-iH_{\rm eff}+\frac{1}{2}\sum_{k\ge 1}\ev{L_k^\dagger L_k}_t)\ket{\psi_t}dt\notag\\
	&+\sum_{k\ge 1}\qty(\frac{L_k\ket{\psi_t}}{\sqrt{\ev{L_k^\dagger L_k}_t}}-\ket{\psi_t})dN_{k,t},
\end{align}
where $\ev{\circ}_t\coloneqq\mel{\psi_t}{\circ}{\psi_t}$ and $dN_{k,t}$ is a stochastic increment that satisfies $\mbb{E}[dN_{k,t}]=\ev{L_k^\dagger L_k}_tdt$ and takes value of $1$ if $k$th jump is detected and $0$ otherwise.

Each stochastic trajectory $\{\ket{\psi_t}\}$ of time duration $\tau$ is uniquely determined by the noise trajectory $\{dN_{k,t}\}$.
Therefore, we can define a time-integrated observable $\phi$ for individual trajectories as
\begin{equation}
	\phi\coloneqq\int_0^\tau\sum_{k\ge 1}dN_{k,t}c_k,
\end{equation}
where $\{c_k\}_{k\ge 1}$ are real counting coefficients.
By this definition, the observable $\phi$ increases by $c_k$ for each $k$th jump occurred. 
For thermodynamically dissipative dynamics, $\phi$ is called a \emph{current} whenever the counting variables are antisymmetric (i.e., $c_k=-c_{k^*}$ for all $k$).
Examples of relevant observables include particle current ($c_k=1$ for absorption and $c_k=-1$ for emission), jump activity ($c_k=1$ for all jumps), and heat flux ($c_k=q_k$ where $q_k$ is heat dissipated from the system to the environment due to $k$th jump).
The average of the counting observable can be analytically calculated as
\begin{equation}
	\ev{\phi}=\tau\sum_{k\ge 1}c_k\tr(L_k\pi L_k^\dagger).
\end{equation}
The higher-order moments of observable $\phi$ can be computed using the method of full counting statistics.
Defining the generating function $G_\tau(u)\coloneqq\tr e^{\mca{L}_u\tau}(\pi)$, the $n$th moment can be calculated as
\begin{equation}\label{eq:C.nth.moments}
	\ev{\phi^n}=(-i\partial_u)^nG_\tau(u)\big|_{u=0},
\end{equation}
where the tilted super-operator $\mca{L}_u$ is given by \cite{Landi.2024.PRXQ}
\begin{equation}
	\mca{L}_u(\varrho)\coloneqq -i[H,\varrho]+\sum_{k\ge 1}\qty(e^{iuc_k}L_k\varrho L_k^\dagger-\{L_k^\dagger L_k,\varrho\}/2).
\end{equation}
It is worth noting that the stochastic unraveling of the GKSL dynamics is not unique, as it depends on the choice of measurement on the environment \cite{Wiseman.2009}. While our primary focus in this study is on quantum jump unraveling, we also demonstrate that similar results can be obtained for quantum diffusion unraveling using the same approach.

\subsection{Entropy production and dynamical activity}
We introduce two relevant quantities for thermodynamically dissipative dynamics and generic dynamics.
The first is irreversible entropy production, which quantifies the degree of time-reversal symmetry breaking \cite{Landi.2021.RMP}.
According to the framework of quantum thermodynamics, entropy production is generally defined as the sum of the entropic changes in the system and the environment \cite{Esposito.2010.NJP}, and its rate is given by
\begin{equation}
	\sigma\coloneqq\sigma_{\rm sys}+\sigma_{\rm env}.
\end{equation}
Here, $\sigma_{\rm sys}$ and $\sigma_{\rm env}$ denote the entropic rate contributed by the system and the environment, respectively. In the generic setup where the system and the environment undergo a unitary evolution, the resulting entropy production can be expressed in terms of a correlation between the system and the environment \cite{Esposito.2010.NJP,Reeb.2014.NJP}. For the weak coupling regime considered here, they are explicitly quantified as \cite{Horowitz.2013.NJP}
\begin{align}
	\sigma_{\rm sys}&\coloneqq-\tr(\dot\varrho_t\ln\varrho_t),\\
	\sigma_{\rm env}&\coloneqq\sum_{k\ge 1}\tr(L_k\varrho_tL_k^\dagger)\Delta s_k.
\end{align}
In the stationary state $\pi$ [i.e., $\mca{L}(\pi)=\mbb{0}$], the entropic contribution from the system vanishes (i.e., $\sigma_{\rm sys}=0$).
Thus, the irreversible entropy production rate is equal to the entropic rate produced by quantum jumps due to the interactions with the environment,
\begin{equation}
	\sigma=\sum_{k\ge 1}\tr(L_k\pi L_k^\dagger)\Delta s_k.
\end{equation}
We can prove that $\sigma\ge 0$, which is nothing but the second law of thermodynamics.
Another crucial quantity is dynamical activity \cite{Maes.2020.PR}, which quantifies the frequency of quantum jumps that occurred in the system and can be explicitly calculated as 
\begin{equation}
	a=\sum_{k\ge 1}\tr(L_k\pi L_k^\dagger).
\end{equation}
Physically, the dynamical activity characterizes the strength of the system's thermalization.
These two quantities ($\sigma$ and $a$) play essential roles in constraining the precision of observables in fluctuating dynamics and other nonequilibrium aspects such as the thermodynamic speed limit \cite{Aurell.2012.JSP,Shiraishi.2018.PRL,Vu.2021.PRL,Vu.2023.PRX}.

\subsection{Symmetrized Liouvillian gap}
Next, we describe another relevant quantity called the symmetrized Liouvillian gap, which plays a crucial role in constraining the fluctuation of observables.
To this end, we introduce the following inner products:
\begin{align}
	\ev{A,B}&\coloneqq\tr(A^\dagger B),\\
	\ev{A,B}_s&\coloneqq\tr(A^\dagger\pi^{s}B\pi^{1-s})
\end{align}
for $s\in[0,1]$.
We also define the norms $\|A\|^2\coloneqq\ev{A,A}$ and $\|A\|_s^2\coloneqq\ev{A,A}_s$ for the notational convenience.
We consider the adjoint time evolution of an operator $A_t$ in the Heisenberg picture, $\dot A_t=\widetilde{\mca{L}}(A_t)$, where the super-operator $\widetilde{\mca{L}}$ is given by
\begin{equation}
\widetilde{\mca{L}}(\circ)\coloneqq i[H,\circ]+\sum_{k\ge 1}\qty(L_k^\dagger\circ L_k-\{L_k^\dagger L_k,\circ\}/2).
\end{equation}
This adjoint super-operator has the following property for any operators $A$ and $B$:
\begin{equation}
	\ev{A,\mca{L}(B)}=\ev{\widetilde{\mca{L}}(A),B}.
\end{equation}
Let $\widetilde{\mca{L}}^*$ be the adjoint super-operator of $\widetilde{\mca{L}}$ with respect to the inner product $\ev{\cdot,\cdot}_s$.
Specifically, $\widetilde{\mca{L}}^*$ is defined such that the following equality is fulfilled for arbitrary operators $A$ and $B$:
\begin{equation}
	\ev{A,\widetilde{\mca{L}}(B)}_s=\ev{\widetilde{\mca{L}}^*(A),B}_s.
\end{equation}
The super-operator $\widetilde{\mca{L}}^*$ can be explicitly expressed as $\widetilde{\mca{L}}^*=\pi^{-s}\mca{L}(\pi^{s}\circ\pi^{1-s})\pi^{s-1}$.
The symmetrized Liouvillian is then defined as \cite{Mori.2023.PRL}
\begin{equation}
	\widetilde{\mca{L}}_s\coloneqq\frac{\widetilde{\mca{L}}+\widetilde{\mca{L}}^*}{2},
\end{equation}
which is self-adjoint with respect to the inner production $\ev{\cdot,\cdot}_s$,
\begin{equation}
	\ev{A,\widetilde{\mca{L}}_s(B)}_s=\ev{\widetilde{\mca{L}}_s(A),B}_s.
\end{equation}
In the literature, $s=0$ and $s=1/2$ are the most studied cases.
Throughout this study, we only consider $s=0$ and $s=1/2$, where $\widetilde{\mca{L}}_s$ possesses some relevant properties.
First, since $\widetilde{\mca{L}}(\mbb{1})=\widetilde{\mca{L}}^*(\mbb{1})=\mbb{0}$, it is evident that $\widetilde{\mca{L}}_s(\mbb{1})=\mbb{0}$.
This means that $\widetilde{\mca{L}}_s$ has a zero eigenvalue corresponding to the eigenvector $\mbb{1}$.
Second, $\widetilde{\mca{L}}_s$ is negative semi-definite with respect to the inner product $\ev{\cdot,\cdot}_s$.
For the $s=0$ case, this was proved in Ref.~\cite{Mori.2023.PRL}.
In Lemma \ref{lem:sym.Lio} (Appendix \ref{app:sL.proof}), we prove the $s=1/2$ case.
Let $\{\lambda_n\}$ be the eigenvalues of $\widetilde{\mca{L}}_s$, which are sorted in the descending order as
\begin{equation}
	0=\lambda_0>\lambda_1\ge\lambda_2\ge\dots\ge\lambda_{d^2-1}.
\end{equation}
The symmetrized Liouvillian gap is defined using the second largest eigenvalue,
\begin{equation}
	g_s\coloneqq-\lambda_1>0.
\end{equation}
The gap $g_s$ quantifies the slowest decay mode of the generator.
As shown later, this gap is essential in constraining the relative fluctuation of observables.

\section{Main results}

\subsection{Quantum thermo-kinetic uncertainty relation}\label{sec:qtkur}
We consider thermodynamically dissipative dynamics, wherein the local detailed balance condition is assumed.
Our first main result is the quantum TKUR, which establishes a lower bound on the relative fluctuation of an arbitrary {\it current} $\phi$.
Applying the quantum Cram{\'e}r-Rao inequality, we prove that the relative fluctuation of current $\phi$ is always lower bounded by entropy production and dynamical activity as 
\begin{equation}\label{eq:main.result.1}
	\frac{F_\phi}{(1+\delta_\phi)^2}\ge\frac{4a}{\sigma^2}\Phi\qty(\frac{\sigma}{2a})^2\ge\max\qty(\frac{2}{\sigma},\frac{1}{a}),
\end{equation}
where $\delta_\phi$ is a contribution to the current average from quantum coherent dynamics [cf.~Eq.~\eqref{eq:cur.avg.qcor}] and $\Phi(x)$ denotes the inverse function of $x\tanh(x)$.
The relation universally holds for arbitrary times, and its second inequality is a consequence of $\Phi(x)\ge\max(\sqrt{x},x)$.
Using the vectorization of operators, the asymptotic long-time value of $\delta_\phi$ can be explicitly given by
\begin{equation}
	\delta_\phi=-\frac{\bvec{\mbb{1}}\widehat{\mca{C}}\widehat{\mca{L}}^+\widehat{\mca{D}}_\ell\kvec{\pi}}{\bvec{\mbb{1}}\widehat{\mca{C}}\kvec{\pi}},
\end{equation}
where $\kvec{A}\coloneqq\sum_{m,n}a_{mn}\ket{m}\otimes\ket{n}$ for $A=\sum_{m,n}a_{mn}\dyad{m}{n}$, $\widehat{\mca{L}}^+$ denotes the Moore-Penrose pseudo-inverse of the vectorized super-operator $\widehat{\mca{L}}$ [cf.~Eq.~\eqref{eq:pseudo.inv}] and
\begin{align}
	\widehat{\mca{D}}_\ell&\coloneqq\sum_{k\ge 1}\ell_k\qty[ L_k\otimes L_k^* - \frac{1}{2} ( L_k^\dagger L_k )\otimes\mbb{1} - \frac{1}{2} \mbb{1}\otimes( L_k^\dagger L_k )^\top],\notag\\
	\widehat{\mca{C}}&\coloneqq\sum_{k\ge 1}c_kL_k\otimes L_k^*.
\end{align}
The real coefficients $\{\ell_k\}$ are given in Eq.~\eqref{eq:pert.dyn}.
The detailed derivation of the result is presented in Appendix \ref{app:proof.qtkur}.
We also demonstrate that a similar relation can be obtained for quantum diffusion unraveling, differing only in the quantum contribution $\delta_\phi$ (see Appendix \ref{app:qTKUR.qdu}).

Remarks on this result are given in order.
(i) First, in the classical limit, we can show that $\delta_\phi=0$ [cf.~Eq.~\eqref{eq:deltaJ.cl}] and the result \eqref{eq:main.result.1} reduces to
\begin{equation}
	F_\phi\ge\frac{4a}{\sigma^2}\Phi\qty(\frac{\sigma}{2a})^2\ge\max\qty(\frac{2}{\sigma},\frac{1}{a}),
\end{equation}
which recovers the TKUR obtained in Ref.~\cite{Vo.2022.JPA} and the conventional TUR \cite{Barato.2015.PRL,Gingrich.2016.PRL} and KUR \cite{Garrahan.2017.PRE,Terlizzi.2019.JPA} for classical Markov jump processes. Therefore, it can be regarded as the quantum TKUR. (ii) Second, inequality \eqref{eq:main.result.1} clearly illustrates the role of quantum coherent dynamics in enhancing the precision of currents. That is, the conventional TKUR is possibly violated if $-2<\delta_\phi<0$. Otherwise, the TKUR is valid whenever $|1+\delta_\phi|\ge 1$. (iii) Third, the bound can be further tightened by considering multiple currents \cite{Moreira.2024.arxiv} and leveraging their correlations, as has been done in the classical case \cite{Dechant.2019.JPA,Vu.2019.PRE.UnderdampedTUR}. (iv) Last, we compare our result with an extant uncertainty relation derived in Ref.~\cite{Vu.2022.PRL.TUR}, which reads
\begin{equation}\label{eq:ext.qtur}
	\frac{F_\phi}{(1+\widetilde{\delta}_\phi)^2}\ge\frac{2}{\sigma+\upsilon},
\end{equation}
where $\upsilon$ denotes a quantum contribution from the quantum coherent dynamics and $\widetilde{\delta}_\phi$ is another correction to the current average.
The quantity $\upsilon$ can be either negative or positive, and there is no definite hierarchical relationship between Eqs.~\eqref{eq:main.result.1} and \eqref{eq:ext.qtur}.
While relation \eqref{eq:ext.qtur} clarifies the role of quantum coherent dynamics in constraining current fluctuations, the separate contributions $\upsilon$ and $\widetilde{\delta}_\phi$ make it intractable to characterize the violation of the TUR in general cases.
In contrast, the new relation \eqref{eq:main.result.1} allows for examining the violation of the TUR through the value of $\delta_\phi$.
Furthermore, by consolidating all quantum contributions into the term $\delta_\phi$, it can derive novel trade-off relations in other contexts as corollaries, as demonstrated below.

\subsubsection*{Power-efficiency trade-off relation for quantum heat engines}
We demonstrate that the quantum TKUR \eqref{eq:main.result.1} can derive a trade-off relation between power and efficiency for quantum steady-state heat engines.
Consider a heat engine simultaneously coupled to two heat baths, one hot at temperature $T_h$ and one cold at temperature $T_c~(<T_h)$.
Let $\phi_h$ be the heat current supplied from the hot heat bath, $\phi_c$ be the heat current absorbed by the cold heat bath.
From the first law of thermodynamics, the power current can be expressed in terms of these heat currents as $P=(\phi_h-\phi_c)/\tau$.
Applying relation \eqref{eq:main.result.1}, we can derive the following trade-off relation between power and efficiency:
\begin{equation}
	P\frac{\eta}{\eta_C-\eta}\frac{T_c(1+\delta_P)^2}{\Delta_P}\le\frac{1}{2}.\label{eq:qua.pe.tradeoff}
\end{equation}
Here, $\eta_C\coloneqq 1-T_c/T_h$ is the Carnot efficiency and $\Delta_P\coloneqq\lim_{\tau\to+\infty}\tau\mvar[P]$.
Notably, unlike the classical trade-off relation between power and efficiency \cite{Pietzonka.2018.PRL}, there exists a quantum contribution $\delta_P$ in inequality \eqref{eq:qua.pe.tradeoff}, which vanishes in the classical limit.
Relation \eqref{eq:qua.pe.tradeoff} implies that achieving the Carnot efficiency at finite power without divergent fluctuation of power necessitates $|1+\delta_P|\ll 1$.
More specifically, $|1+\delta_P|$ should vanish at the same order of $\sqrt{\eta_C-\eta}$ as $\eta\to\eta_C$.
This provides insights into the design of efficient heat engines operating at the boundary of the fundamental limitations.

\subsection{Quantum inverse uncertainty relation}\label{sec:qiur}
Next, we consider generic dynamics \eqref{eq:Lindblad.dyn} without the assumption of local detailed balance. 
We deal with an {\it arbitrary} counting observable $\phi$, including currents studied in the previous subsection \ref{sec:qtkur}.
Our second main result is the quantum inverse uncertainty relation, which sets an upper bound on the relative fluctuation of observable $\phi$, expressed as (see Appendix \ref{app:proof.qiur} for the proof)
\begin{equation}\label{eq:main.result.2}
	F_\phi\le\frac{\ev{J_2,\pi}}{\ev{J_1,\pi}^2}\qty(1+\frac{2\kappa}{g_s}).
\end{equation}
Here, $J_n\coloneqq\sum_{k\ge 1}c_k^nL_k^\dagger L_k$ ($n=1,2$) are self-adjoint operators defined in terms of jump operators and counting coefficients, and $\kappa$ is given by
\begin{align}
	\kappa&\coloneqq\frac{\|J_1-\ev{J_1,\pi}\mbb{1}\|_s\|\pi^{-s}J_\pi\pi^{s-1}-\ev{J_1,\pi}\mbb{1}\|_s}{\ev{J_2,\pi}},\\
	J_\pi&\coloneqq\sum_{k\ge 1}c_kL_k\pi L_k^\dagger.
\end{align}
Physically, $J_1$ can be identified as the observable operator as its expected value with respect to the stationary state $\pi$ yields the observable average; that is, $\ev{J_1,\pi}=\tau^{-1}\ev{\phi}$. On the other hand, $J_2$ can be regarded as the instantaneous-fluctuation operator as $\ev{J_2,\pi}=\lim_{\tau\to 0}\tau^{-1}\mvar[\phi]$. These immediately derive
\begin{equation}
	\frac{\ev{J_2,\pi}}{\ev{J_1,\pi}^2}=\lim_{\tau\to 0}\frac{\tau\mvar[\phi]}{\ev{\phi}^2}.
\end{equation}
Thus, ${\ev{J_2,\pi}}/{\ev{J_1,\pi}^2}$ can be interpreted as the instantaneous relative fluctuation of the observable.
Additionally, it can be observed that $\kappa$ vanishes when either $J_1\propto\mds{1}$ or $J_\pi\propto\pi$, indicating that $\kappa$ quantifies the deviation of the observable operator from the identity operator.
Inequality \eqref{eq:main.result.2}, valid for $s\in\{0,1/2\}$ and for arbitrary times, reveals that the precision of observables is fundamentally constrained by the symmetrized spectral gap and the instantaneous relative fluctuation.
In general, there is no strict hierarchical relationship between the bounds for $s=0$ and $s=1/2$; however, they converge and become identical in the classical limit. A similar bound for quantum diffusion unraveling is presented in Appendix \ref{app:qiur.qdu}.
Combined with the quantum TKUR \eqref{eq:main.result.1}, this bound offers a comprehensive picture for understanding the precision of currents in finite-time processes.
Furthermore, this result extends the classical findings of Ref.~\cite{Smith.2023.PRL} for Markov jump processes, which relied on concentration techniques, into the quantum domain. 
While some quantum bounds on the moment generating function utilizing the $s=1/2$ spectral gap have been established for simple counting processes \cite{Girotti.2023.AHP} and diffusive processes \cite{Benoist.2022.Q}, our approach takes a different route by directly bounding the second moment and accounting for the spectral gap in both $s=0$ and $s=1/2$ cases. 
This provides a more comprehensive and versatile framework for exploring the precision constraints in quantum systems.

\subsection{Quantum response kinetic uncertainty relation}\label{sec:qrkur}
Last, we investigate the static response of observables to kinetic perturbations.
In contrast to the previous subsections, we focus on observables expressed in a general form $f(\vb*{\phi})$, where $f$ is an arbitrary function and $\vb*{\phi}=[\phi_1,\dots,\phi_N]^\top$ is a vector of arbitrary counting observables.
Notably, when $f(x)=x$ and $N=1$, this observable reduces to the conventional counting observable.
We consider the case where the Hamiltonian and the jump operators are parameterized by a control variable $\epsilon$.
That is, the dependence of each jump operator on $\epsilon$ is explicitly given by
\begin{equation}
	L_k=e^{\omega_k(\epsilon)/2}V_k,
\end{equation}
where $\omega_k(\epsilon)$ is a function of $\epsilon$ and $V_k$ is independent of $\epsilon$.
In contrast, the Hamiltonian can depend on $\epsilon$ in an arbitrary manner; that is, $H$ can take any form $H(\epsilon)$.
Examples of $\epsilon$ include measurement amplitude, energy-level spacing, reservoir temperature, and the coupling strength between the system and the reservoir.

As the third main result, we obtain the following response kinetic uncertainty relation for arbitrary observables $f(\vb*{\phi})$ (see Appendix \ref{app:proof.qrkur} for the proof):
\begin{equation}\label{eq:main.result.3}
	\frac{\|\nabla\ev{f(\vb*{\phi})}\|_1^2}{\mvar[f(\vb*{\phi})]}\le\tau a.
\end{equation}
Here, $\nabla\ev{f(\vb*{\phi})}\coloneqq[d_{\omega_k}\ev{f(\vb*{\phi})}]_k^\top$ is the vector of the static responses of the observable to each jump perturbation $\omega_k\to\omega_k+\delta\omega_k$ and $\|\vb*{x}\|_1\coloneqq\sum_n|x_n|$ denotes the $1$-norm.
More precisely, the total derivative symbol $d$ is defined as $d_w\ev{\circ}\coloneqq\lim_{\delta w\to 0}[\ev{\circ}_{w+\delta w}-\ev{\circ}_w]/\delta w$, where the initial state under perturbations remains fixed at $\pi$.
Relation \eqref{eq:main.result.3}, which holds for arbitrary times, indicates that the precision of observable response to small perturbations is always bounded by dynamical activity.
While our focus here is on the stationary state, this result can also be extended to the transient regime, retaining the same structure (see Appendix \ref{app:res.kur.tran.gen} for the detailed derivation). 
Specifically, we can prove that the following relation holds for any operational time and arbitrary initial quantum states:
\begin{equation}
	\frac{\|\nabla\ev{f(\vb*{\phi})}\|_1^2}{\mvar[f(\vb*{\phi})]}\le\mca{A}_\tau,
\end{equation}
where $\mca{A}_\tau\coloneqq\int_0^\tau\dd{t}\sum_{k\ge 1}\tr(L_k\varrho_tL_k^\dagger)$ is the dynamical activity, quantifying the average number of quantum jumps over the duration $\tau$.
This generalization demonstrates that the trade-off persists across both stationary and transient processes, thereby significantly broadening the applicability of the derived bounds to encompass a wider range of nonequilibrium quantum dynamics.
Furthermore, the bound also holds for quantum diffusion unraveling, as its derivation is independent of the unraveling method.
In the following, we demonstrate its twofold applications.

\begin{figure*}[t]
\centering
\includegraphics[width=1\linewidth]{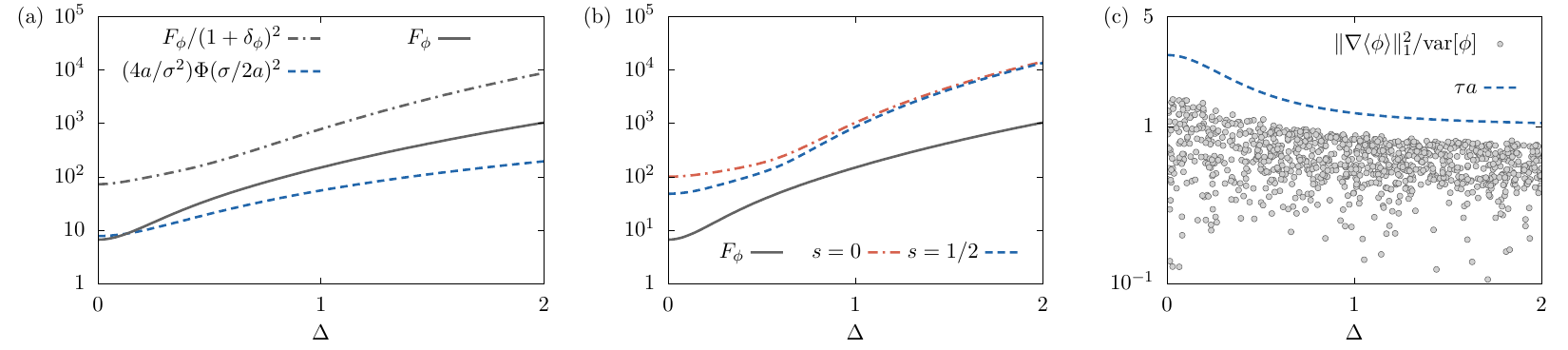}
\protect\caption{Numerical illustration of the main results in the three-level maser. (a) The validity of the quantum TKUR \eqref{eq:main.result.1} and the violation of the classical version. The solid and dashed-dotted lines depict $F_\phi$ and $F_\phi/(1+\delta_\phi)^2$, respectively, whereas the dashed line represents the thermo-kinetic lower bound $(4a/\sigma^2)\Phi(\sigma/2a)^2$. 
(b) Demonstration of the quantum inverse uncertainty relation \eqref{eq:main.result.2}. The upper bounds computed with $s=0$ and $s=1/2$ are plotted using the dashed-dotted and dashed lines, respectively.
(c) Demonstration of the quantum response kinetic uncertainty relation \eqref{eq:main.result.3}. The circles and the dashed line depict the response precision $\|\nabla\ev{\phi}\|_1^2/\mvar[\phi]$ and the upper bound of dynamical activity. $\Delta$ is varied in the range $[0,2]$, while other parameters are fixed as $\gamma_h=0.1$, $\gamma_c=2$, $n_h=5$, $n_c=0.02$, $\Omega=0.15$, and $\tau=10$.}\label{fig:Maser}
\end{figure*}

First, it can derive an upper bound on the observable response to a small change in the parameter $\epsilon$.
To this end, note that $|d_{\omega_k}\ev{f(\vb*{\phi})}|\ge|d_\epsilon\omega_k(\epsilon)d_{\omega_k}\ev{f(\vb*{\phi})}|/\omega_{\rm max}$, where $\omega_{\rm max}\coloneqq\max_k|d_\epsilon\omega_k(\epsilon)|$.
Using this fact and considering the case where the Hamiltonian is independent of $\epsilon$, we obtain
\begin{align}
	\|\nabla\ev{f(\vb*{\phi})}\|_1&\ge\frac{1}{\omega_{\rm max}}|\sum_{k\ge 1}d_\epsilon\omega_k(\epsilon)d_{\omega_k}\ev{f(\vb*{\phi})}|\notag\\
	&=\frac{|d_\epsilon\ev{f(\vb*{\phi})}|}{\omega_{\rm max}}.
\end{align}
Consequently, we can show that the response of observables to the parameter perturbation is bounded by dynamical activity and its fluctuation,
\begin{equation}\label{eq:res.per.kur}
	[d_\epsilon\ev{f(\vb*{\phi})}]^2\le \tau \omega_{\rm max}^2\mvar[f(\vb*{\phi})]a.
\end{equation}
Inequality \eqref{eq:res.per.kur} implies that achieving a large response of an observable requires either significant fluctuations or high dynamical activity in the system.
Additionally, this result provides a quantum generalization of the classical response uncertainty relation reported in Refs.~\cite{Aslyamov.2024.arxiv,Liu.2024.arxiv} for Markov jump processes, offering a broader framework that goes beyond counting observables.

Second, relation \eqref{eq:main.result.3} can recover the KUR in the classical limit.
This can be verified from the fact that $\|\nabla\ev{\phi}\|_1\ge|\ev{\phi}|$ (see Appendix \ref{app:ckur.recover} for the proof), which immediately yields the KUR \cite{Garrahan.2017.PRE,Terlizzi.2019.JPA},
\begin{equation}
	\frac{\ev{\phi}^2}{\mvar[\phi]}\le\tau a.
\end{equation}
In other words, relation \eqref{eq:main.result.3} is quantitatively tighter than the KUR.

\section{Numerical demonstration}
\subsection{Three-level maser engine}
We exemplify the main results in a three-level maser engine \cite{Scovil.1959.PRL}, which is resonantly modulated by an external electric field and simultaneously coupled to a hot and a cold heat bath.
The maser can operate as a heat engine or refrigerator depending on the parameter regime.
The Hamiltonian and jump operators are explicitly given by
\begin{align}
H_t&=H_0+V_t,\\
L_1&= \sqrt{\gamma_hn_h}\sigma_{31},~L_{1^*}=\sqrt{\gamma_h(n_h+1)}\sigma_{13},\\
L_2&=\sqrt{\gamma_cn_c}\sigma_{32},~L_{2^*}=\sqrt{\gamma_c(n_c+1)}\sigma_{23}.
\end{align}
Here, $H_0=\omega_1\sigma_{11}+\omega_2\sigma_{22}+\omega_3\sigma_{33}$ is the bare Hamiltonian, $V_t=\Omega\qty( e^{i\omega_0t}\sigma_{12}+ e^{-i\omega_0t}\sigma_{21})$ is the external classical field, and $\sigma_{ij}=\dyad{i}{j}$ for some basis $\{\ket{i}\}$.
To remove the time dependence of the full Hamiltonian, we consider operators in an appropriate rotating frame $X\to\tilde{X}= U_t^\dagger XU_t$ \cite{Boukobza.2007.PRL}, where $U_t=e^{-i\bar{H}t}$ and $\bar{H}=\omega_1\sigma_{11}+(\omega_1+\omega_0)\sigma_{22}+\omega_3\sigma_{33}$.
In this rotating frame, the GKSL master equation reads \cite{fnt1}
\begin{equation}\label{eq:rot.frame.Ham}
\dot{\tilde{\varrho}}_t = -i[ H,\tilde{\varrho}_t ] +\sum_{k\ge 1}\qty(L_k\tilde{\varrho}_tL_k^\dagger-\acomm{L_k^\dagger L_k}{\tilde{\varrho}_t}/2),
\end{equation}
where $H=-\Delta \sigma_{22}+\Omega( \sigma_{12}+\sigma_{21} )$ and $\Delta=\omega_0+\omega_1-\omega_2$.

We consider a current $\phi$ with $\vb*{c}=[1,-1,-1,1]^\top$, which is proportional to the net number of cycles.
To examine the validity of the main results, we vary only $\Delta$ while fixing other parameters.
The relative fluctuation of current $\phi$ is numerically calculated using the method of full counting statistics.
All the numerical results are plotted in Fig.~\ref{fig:Maser}.
As shown in Fig.~\ref{fig:Maser}(a), the relative fluctuation $F_\phi$ increases with increasing $\Delta$.
Notably, the quantum TKUR \eqref{eq:main.result.1} holds universally, as the line representing $F_\phi/(1+\delta_\phi)^2$ always lies strictly above the lower bound $(4a/\sigma^2)\Phi(\sigma/2a)^2$.
In contrast, when the quantum correction $\delta_\phi$ is excluded, the relative fluctuation $F_\phi$ falls below the lower bound for $\Delta\ll\Omega$, signaling a violation of the classical TKUR.

The quantum inverse uncertainty relation \eqref{eq:main.result.2} is numerically verified in Fig.~\ref{fig:Maser}(b) for the same current $\phi$.
The upper bound is calculated for both the $s=0$ and $s=1/2$ cases.
As shown, the relative fluctuation is always constrained by the upper bound.
Notably, the bound for $s=1/2$ is tighter than that for $s=0$ in this scenario.

Finally, we validate the quantum response kinetic uncertainty relation \eqref{eq:main.result.3}. To calculate the response precision, we sample the observable $\phi$, with each element of $\vb*{c}$ randomly chosen from the range $[-1,1]$. As illustrated in Fig.~\ref{fig:Maser}(c), while the dynamical activity decreases, it consistently acts as an upper bound for the response precision of all sampled observables.

\begin{figure*}[t]
\centering
\includegraphics[width=1\linewidth]{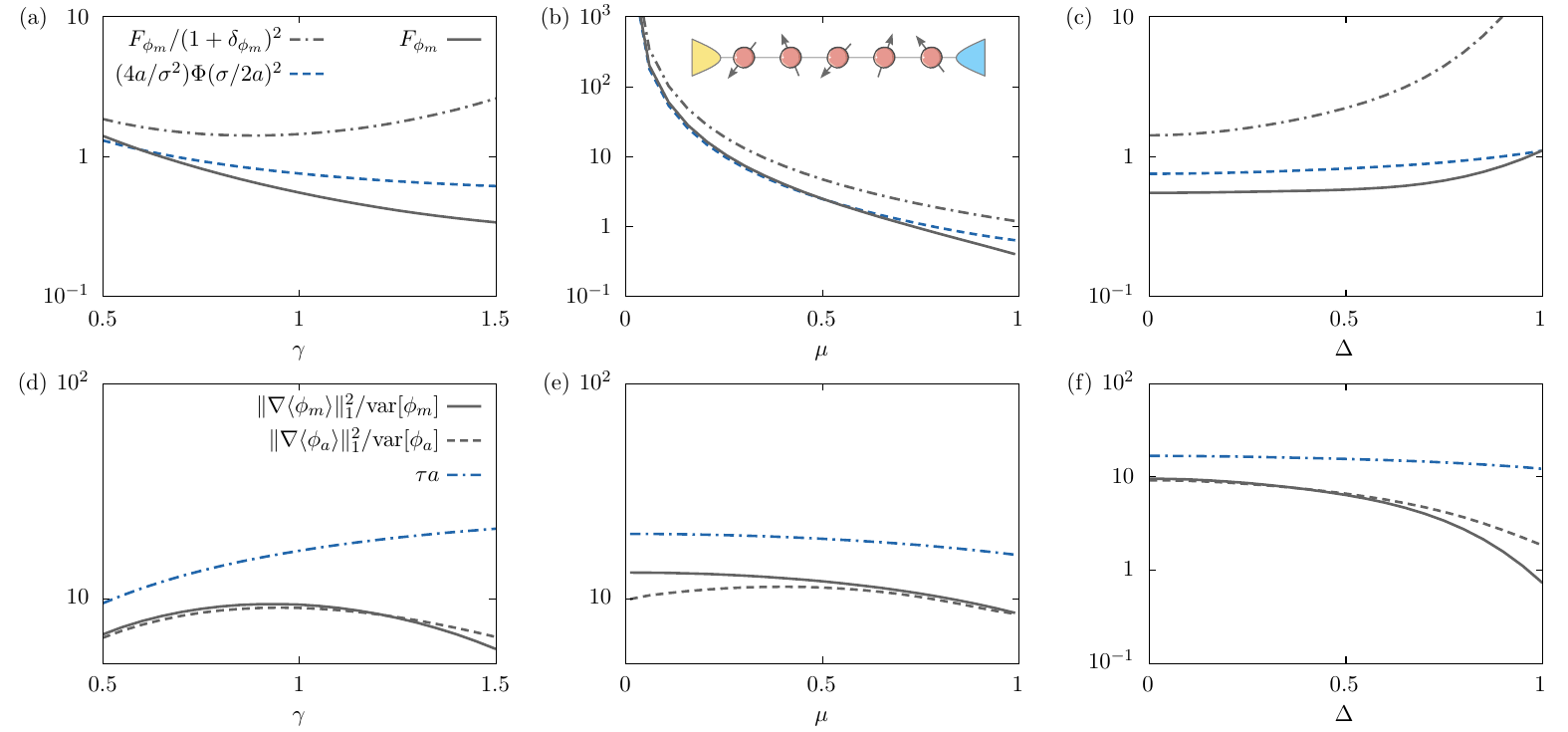}
\protect\caption{Numerical illustration of the main results in the XXZ spin chain with $L=5$ sites. (Top panel) The quantum TKUR, represented by $F_{\phi_m}/(1+\delta_{\phi_m})^2$, and the conventional TKUR, $F_{\phi_m}$, are shown as the dashed-dotted and solid lines, respectively, while the lower bound $(4a/\sigma^2)\Phi(\sigma/2a)^2$ is depicted by the dashed line. 
(Bottom panel) The response precision of current $\phi_m$ and observable $\phi_a$ are represented by the solid and dashed lines, respectively, while the upper bound $\tau a$ is depicted by the dashed-dotted line.
In each corresponding figure, the fixed parameters are $\gamma=1$, $\mu=0.9$, $\Delta=0.1$, and $\tau=10$.}\label{fig:SpinChain}
\end{figure*}

\subsection{Boundary-driven XXZ spin chain}
Next, we illustrate our results in a quantum many-body spin system, specifically the spin-$1/2$ XXZ chain on one-dimensional lattices.
The boundary-driven XXZ spin chain is a paradigmatic model for studying nonequilibrium quantum dynamics, offering profound insights into transport phenomena, steady-state properties, and the role of integrability in dissipative systems \cite{Landi.2022.RMP}.
By coupling the spin chain to external reservoirs at its boundaries, this model provides a controlled platform to investigate how energy, spin, or particle currents emerge and evolve under nonequilibrium conditions.

We consider a spin chain consisting of $L$ sites with the Hamiltonian given by 
\begin{equation}
	H_{\rm XXZ}=\sum_{n=1}^{L-1}(\sigma_n^x\sigma_{n+1}^x+\sigma_n^y\sigma_{n+1}^y+\Delta\sigma_n^z\sigma_{n+1}^z),
\end{equation}
where $\sigma_n^a$ $(a\in\{x,y,z\})$ are the standard Pauli matrices at site $n$ and $\Delta$ is the anisotropy parameter.
The spin chain is coupled to reservoirs at only the first and last sites, which drive the system out of equilibrium.
The dissipative operators are specified as
\begin{align}
	L_{1}&=\sqrt{\gamma(1+\mu)}\sigma_1^+,~L_{1^*}=\sqrt{\gamma(1-\mu)}\sigma_1^-,\\
	L_{2}&=\sqrt{\gamma(1-\mu)}\sigma_L^+,~L_{2^*}=\sqrt{\gamma(1+\mu)}\sigma_L^-,
\end{align}
where $\sigma_n^\pm=(\sigma_n^x\pm i\sigma_n^y)/2$, $\gamma$ describes the coupling strength to reservoirs, and $\mu\in[0,1]$ reflects the strength of incoherent driving at the boundaries.
It has been shown that the system possesses a unique nonequilibrium stationary state \cite{Prosen.2012.PS}.
Under symmetric driving conditions (i.e., $\mu=0$), the steady state corresponds to a thermal state at infinite temperature, $\pi\propto\mbb{1}$.
In contrast, for asymmetric driving (i.e., $\mu > 0$), a persistent nonzero current emerges in the steady state.
This setup allows us to analyze transport properties and serves as an ideal platform to verify the derived trade-off relations for precision, response, and thermodynamic costs in quantum systems.

We consider two observables: the magnetization current $\phi_m$ and the activity observable $\phi_a$.
Since the Hamiltonian $H_{\rm XXZ}$ conserves the total magnetization $M=\sum_{n=1}^L\sigma_n^z$ (i.e., $[H_{\rm XXZ},M]=\mbb{0}$), the time evolution of the total magnetization $\ev{M}_t\coloneqq\tr(M\varrho_t)$ can be described as
\begin{equation}
	\frac{d\ev{M}_t}{dt}=\iota_1(t)+\iota_L(t),
\end{equation}
where $\iota_{1/L}(t)$ denote the magnetization fluxes injected from the reservoirs at the boundaries.
At the stationary state, the fluxes satisfy $\iota_1^{\rm ss}+\iota_L^{\rm ss}=0$, and $\iota_1^{\rm ss}$ can be explicitly expressed as
\begin{equation}
	\iota_1^{\rm ss}=2\qty[\tr(L_1\pi L_1^\dagger)-\tr(L_{1^*}\pi L_{1^*}^\dagger)].
\end{equation}
Therefore, the stochastic magnetization current $\phi_m$ can be specified using the vector of counting coefficients $\vb*{c}_m=[2,-2,0,0]^\top$.
On the other hand, the activity observable is defined by choosing the counting coefficients $\vb*{c}_a=[1,1,0,0]^\top$.
This observable simply counts the number of spin flips occurring at the first site of the spin chain.

We numerically verify the quantum TKUR \eqref{eq:main.result.1} and the response kinetic uncertainty relation \eqref{eq:main.result.3} in the spin chain with $L=5$ sites. One of the parameters $\gamma$ (coupling strength to reservoirs), $\mu$ (boundary driving asymmetry), and $\Delta$ (anisotropy) is varied, while the others are fixed.
As shown in the top panel of Fig.~\ref{fig:SpinChain}, the quantum TKUR remains valid for the magnetization current across all parameter ranges. 
Moreover, the bound accurately captures the behavior of the magnetization precision as the dissipative strength is varied in Fig.~\ref{fig:SpinChain}(b).
In contrast, the conventional TKUR can be significantly violated. Specifically, the current precision is enhanced by quantum coherent dynamics when either the coupling strength to the reservoirs or the boundary driving asymmetry increases. Notably, as illustrated in Fig.~\ref{fig:SpinChain}(c), the conventional TKUR is most strongly violated in the regime $0\le\Delta\le 1$, where transport is known to be ballistic \cite{Landi.2022.RMP}.
These results demonstrate the critical role of quantum effects in enhancing precision and reveal their intricate interplay with thermodynamic quantities in boundary-driven systems.
In the bottom panel of Fig.~\ref{fig:SpinChain}, the response precision of both the magnetization current and the activity observable is consistently constrained by the dynamical activity, confirming the validity of the quantum response kinetic uncertainty relation \eqref{eq:main.result.3}.

\section{Summary and discussion}
In this study, we derived fundamental bounds on the precision and response of trajectory observables for Markovian dynamics. Our first result is the quantum TKUR \eqref{eq:main.result.1}, which provides a lower bound on the relative fluctuations of currents in terms of entropy production and dynamical activity. Unlike the classical version, the new relation includes a quantum correction term $\delta_\phi$, highlighting the role of quantum coherent dynamics in enhancing the precision of observables. Applying this result to quantum heat engines, we derived a novel trade-off relation between power and efficiency, offering insights into designing quantum engines capable of achieving finite power at the Carnot efficiency without divergent fluctuations. A promising direction for future research is exploring this possibility through numerical optimization to identify scenarios where maximal efficiency and minimal quantum correction $|1+\delta_P|$ are achieved at finite power. It is also relevant to derive an analogous relation to \eqref{eq:main.result.1} for the first passage time of currents \cite{Vu.2022.PRL.TUR,Hasegawa.2022.PRE,Kewming.2024.PRA,Smith.2024.arxiv}.

The second result is the quantum inverse uncertainty relation \eqref{eq:main.result.2}, which sets an upper bound on the relative fluctuations of observables in terms of their instantaneous fluctuation and spectral gap. This result complements the first, as combining them yields a sandwich bound for current fluctuations. An open question is whether the upper bound can be refined to include contributions from entropy production.

Our third result is the quantum response kinetic uncertainty relation \eqref{eq:main.result.3}, which provides an upper bound on the precision of the static response of general observables to small parameter changes, expressed in terms of dynamical activity. The result can be experimentally verified, particularly in the context of quantum continuous measurement \cite{Wiseman.2009}, where all the quantities in the bound are measurable. Exploring its connection with the quantum KUR \cite{Hasegawa.2020.PRL,Vu.2022.PRL.TUR,Prech.2025.PRL} and extending the framework to derive a trade-off relation in terms of entropy production for current-type observables remains an intriguing avenue for future research.

\emph{Note added.}---After completing this work, we became aware that Ref.~\cite{Kwon.2024.arxiv} had obtained a similar result to relation \eqref{eq:main.result.3}.

\begin{acknowledgments}
The author thanks Keiji Saito for fruitful discussions.
This work was supported by JSPS KAKENHI Grant No.~JP23K13032.
\end{acknowledgments}

\onecolumngrid

\appendix

\section{Properties of the symmetrized Liouvillian $\widetilde{\mca{L}}_s$}\label{app:sL.proof}
\begin{lemma}\label{lem:sym.Lio}
The symmetrized Liouvillian $\widetilde{\mca{L}}_s$ is negative semi-definite with respect to the inner product $\ev{\cdot,\cdot}_s$ for $s=1/2$. That is, the following inequality holds true for any operator $A$:
\begin{equation}
	\ev{A,\widetilde{\mca{L}}_s(A)}_s\le 0.
\end{equation}
Furthermore, the spectral gap $g_s$ can be expressed in the following variational form:
\begin{equation}
	g_s=-\max_{\ev{A,\mbb{1}}_s=0}\frac{\ev{A,\widetilde{\mca{L}}_s(A)}_s}{\ev{A,A}_s}.
\end{equation}
\end{lemma}
\begin{proof}
Define $\mca{L}_s\coloneqq(\mca{L}+\mca{L}^*)/2$, where $\mca{L}^*(\circ)\coloneqq\pi^{1/2}\widetilde{\mca{L}}(\pi^{-1/2}\circ\pi^{-1/2})\pi^{1/2}$.
We can show that $\mca{L}^*$ is an adjoint super-operator of $\widetilde{\mca{L}}^*$ with respect to the inner product $\ev{\cdot,\cdot}$ as follows:
\begin{align}
	\ev{\widetilde{\mca{L}}^*(A),B}&=\tr{\pi^{-1/2}\mca{L}(\pi^{1/2}A\pi^{1/2})^\dagger\pi^{-1/2}B}\notag\\
	&=\ev{\mca{L}(\pi^{1/2}A\pi^{1/2}),\pi^{-1/2}B\pi^{-1/2}}\notag\\
	&=\ev{\pi^{-1/2}B\pi^{-1/2},\mca{L}(\pi^{1/2}A\pi^{1/2})}^*\notag\\
	&=\ev{\widetilde{\mca{L}}(\pi^{-1/2}B\pi^{-1/2}),\pi^{1/2}A\pi^{1/2}}^*\notag\\
	&=\ev{A,\pi^{1/2}\widetilde{\mca{L}}(\pi^{-1/2}B\pi^{-1/2})\pi^{1/2}}\notag\\
	&=\ev{A,\mca{L}^*(B)}.
\end{align}
Therefore, $\mca{L}_s$ is an adjoint super-operator of $\widetilde{\mca{L}}_s$ with respect to the inner product $\ev{\cdot,\cdot}$ because
\begin{align}
	\ev{\widetilde{\mca{L}}_s(A),B}&=\frac{1}{2}\qty[\ev{\widetilde{\mca{L}}(A),B}+\ev{\widetilde{\mca{L}}^*(A),B}]\notag\\
	&=\frac{1}{2}\qty[\ev{A,\mca{L}(B)}+\ev{A,\mca{L}^*(B)}]\notag\\
	&=\ev{A,\mca{L}_s(B)}.
\end{align}
From the definition $\mca{L}_s=(\mca{L}+\mca{L}^*)/2$, we can explicitly express it as
\begin{align}
	2\mca{L}_s(\circ)&=G\circ+\circ G^\dagger+\sum_{k\ge 1}(L_k\circ L_k^\dagger + \widetilde{L}_k\circ\widetilde{L}_k^\dagger),
\end{align}
where we define $\widetilde{H}\coloneqq\pi^{1/2}H\pi^{-1/2}$, $\widetilde{L}_k\coloneqq\pi^{1/2}L_k^\dagger\pi^{-1/2}$, and
\begin{equation}
	G\coloneqq-iH+i\widetilde{H}-\frac{1}{2}\sum_{k\ge 1}\qty(L_k^\dagger L_k+\pi^{1/2}L_k^\dagger L_k\pi^{-1/2}).
\end{equation}
Since $\mca{L}(\pi)=\mbb{0}$, we have
\begin{align}
	G+G^\dagger&=i(\pi^{1/2}H\pi^{-1/2}-\pi^{-1/2}H\pi^{1/2})-\sum_{k\ge 1}L_k^\dagger L_k-\frac{1}{2}\sum_{k\ge 1}(\pi^{1/2}L_k^\dagger L_k\pi^{-1/2}+\pi^{-1/2}L_k^\dagger L_k\pi^{1/2})\notag\\
	&=-\sum_{k\ge 1}L_k^\dagger L_k+\pi^{-1/2}\qty[-i(H\pi-\pi H)-\frac{1}{2}\sum_{k\ge 1}\{L_k^\dagger L_k,\pi\}]\pi^{-1/2}\notag\\
	&=-\sum_{k\ge 1}L_k^\dagger L_k+\pi^{-1/2}\qty(\mca{L}(\pi)-\sum_{k\ge 1}L_k\pi L_k^\dagger)\pi^{-1/2}\notag\\
	&=-\sum_{k\ge 1}(L_k^\dagger L_k+\widetilde{L}_k^\dagger\widetilde{L}_k).
\end{align}
The operator $G$ can be decomposed as the sum of self-adjoint operators as $G=A-iH_s$, where $A=(G+G^\dagger)/2$ and $H_s=(G^\dagger-G)/(2i)$ are self-adjoint operators.
Consequently, it can be easily verified that $\mca{L}_s$ is the generator of a completely positive trace-preserving map,
\begin{align}
	2\mca{L}_s(\circ)&=-i[H_s,\circ]+\sum_{k\ge 1}\qty(L_k\circ L_k^\dagger-\frac{1}{2}\{L_k^\dagger L_k,\circ\}+\widetilde{L}_k\circ \widetilde{L}_k^\dagger-\frac{1}{2}\{\widetilde{L}_k^\dagger\widetilde{L}_k,\circ\}).
\end{align}
Thus, all the eigenvalues of $\mca{L}_s$ have nonpositive real parts.
Additionally, we can show that $\mca{L}_s'\coloneqq\pi^{-1/4}\mca{L}_s(\pi^{1/4}\circ\pi^{1/4})\pi^{-1/4}$ is a self-adjoint super-operator, that is,
\begin{equation}
	\ev{\mca{L}_s'(A),B}=\ev{A,\mca{L}_s'(B)}.
\end{equation}
Indeed, it can be proved as follows:
\begin{align}
	\ev{\mca{L}_s'(A),B}&=\frac{1}{2}\ev{\pi^{-1/4}\mca{L}(\pi^{1/4}A\pi^{1/4})\pi^{-1/4}+\pi^{-1/4}\mca{L}^*(\pi^{1/4}A\pi^{1/4})\pi^{-1/4},B}\notag\\
	&=\frac{1}{2}\ev{\pi^{-1/4}\mca{L}(\pi^{1/4}A\pi^{1/4})\pi^{-1/4}+\pi^{1/4}\widetilde{\mca{L}}(\pi^{-1/4}A\pi^{-1/4})\pi^{1/4},B}\notag\\
	&=\frac{1}{2}\ev{\mca{L}(\pi^{1/4}A\pi^{1/4}),\pi^{-1/4}B\pi^{-1/4}}+\ev{\widetilde{\mca{L}}(\pi^{-1/4}A\pi^{-1/4}),\pi^{1/4}B\pi^{1/4}}\notag\\
	&=\frac{1}{2}\ev{A,\pi^{1/4}\widetilde{\mca{L}}(\pi^{-1/4}B\pi^{-1/4})\pi^{1/4}}+\ev{A,\pi^{-1/4}\mca{L}(\pi^{1/4}B\pi^{1/4})\pi^{-1/4}}\notag\\
	&=\ev{A,\mca{L}_s'(B)}.\label{eq:Lsp.sym}
\end{align}
For any super-operator $\mca{S}$ using calligraphic fonts, we denote by $\widehat{\mca{S}}$ its corresponding vectorization representation, which vectorizes operators as
\begin{equation}
A=\sum_{m,n}a_{mn}\dyad{m}{n}\rightarrow \kvec{A}=\sum_{m,n}a_{mn}\ket{m}\otimes\ket{n}.
\end{equation}
By algebraic calculations, one can show that $\kvec{AB}=(A\otimes\mbb{1})\kvec{B}$ and $\kvec{BA}=(\mbb{1}\otimes A^\top)\kvec{B}$.
It is evident from Eq.~\eqref{eq:Lsp.sym} that $\widehat{\mca{L}}_s'$ is a Hermitian matrix and its spectral decomposition is given by
\begin{equation}
	\widehat{\mca{L}}_s'=\sum_{i=0}^{d^2-1}\lambda_i\kvec{r_i'}\bvec{r_i'},
\end{equation}
where $\{\lambda_i\}$ are real eigenvalues and the eigenvectors $\{\kvec{r_i'}\}$ form an orthogonal and complete basis $\sum_{i=0}^{d^2-1}\kvec{r_i'}\bvec{r_i'}=\mbb{1}$.
Since both $\widehat{\mca{L}}_s$ and $\widehat{\mca{L}}_s'$ have the same eigenvalues, $\{\lambda_i\}$ are real nonpositive numbers sorted in the descending order $\lambda_0\ge\dots\ge\lambda_{d^2-1}$.
Note that $\lambda_0=0$ and $\kvec{r_0'}=\kvec{\pi^{1/2}}$. 
From the relation $\widehat{\mca{L}}_s'=[\pi^{-1/4}\otimes(\pi^{-1/4})^\top]\widehat{\mca{L}}_s[\pi^{1/4}\otimes(\pi^{1/4})^\top]$, the explicit form of $\widehat{\mca{L}}_s$ can be calculated as
\begin{equation}
	\widehat{\mca{L}}_s=[\pi^{1/4}\otimes(\pi^{1/4})^\top]\widehat{\mca{L}}_s'[\pi^{-1/4}\otimes(\pi^{-1/4})^\top]=\sum_{i=0}^{d^2-1}\lambda_i[\pi^{1/4}\otimes(\pi^{1/4})^\top]\kvec{r_i'}\bvec{r_i'}[\pi^{-1/4}\otimes(\pi^{-1/4})^\top].
\end{equation}
Since $\widetilde{\mca{L}}_s$ is the adjoint super-operator of $\mca{L}_s$, we obtain
\begin{equation}
	\widehat{\widetilde{\mca{L}}}_s=\widehat{\mca{L}}_s^\dagger=\sum_{i=0}^{d^2-1}\lambda_i[\pi^{-1/4}\otimes(\pi^{-1/4})^\top]\kvec{r_i'}\bvec{r_i'}[\pi^{1/4}\otimes(\pi^{1/4})^\top].
\end{equation}
Using this form, we can prove the nonnegativity of the inner product $\ev{A,\widetilde{\mca{L}}_s(A)}_s$ as follows:
\begin{align}
	\ev{A,\widetilde{\mca{L}}_s(A)}_s&=\ev{A,\pi^{1/2}\widetilde{\mca{L}}_s(A)\pi^{1/2}}\notag\\
	&=\bvec{A}[\pi^{1/2}\otimes(\pi^{1/2})^\top]\widehat{\widetilde{\mca{L}}}_s\kvec{A}\notag\\
	&=\sum_{i=0}^{d^2-1}\lambda_i\bvec{A}[\pi^{1/4}\otimes(\pi^{1/4})^\top]\kvec{r_i'}\bvec{r_i'}[\pi^{1/4}\otimes(\pi^{1/4})^\top]\kvec{A}\notag\\
	&=\sum_{i=0}^{d^2-1}\lambda_i|\bvec{A}[\pi^{1/4}\otimes(\pi^{1/4})^\top]\kvec{r_i'}|^2\notag\\
	&=\sum_{i=1}^{d^2-1}\lambda_i|\bvec{A}[\pi^{1/4}\otimes(\pi^{1/4})^\top]\kvec{r_i'}|^2\le 0.
\end{align}
Finally, we will prove that
\begin{equation}
	g_s=-\max_{\ev{A,\mbb{1}}_s=0}\frac{\ev{A,\widetilde{\mca{L}}_s(A)}_s}{\ev{A,A}_s}.
\end{equation}
To this end, note that $\ev{A,\mbb{1}}_s=\ev{A,\pi}=\ev{A,\pi^{1/4}\pi^{1/2}\pi^{1/4}}=\bvec{A}[\pi^{1/4}\otimes(\pi^{1/4})^\top]\kvec{\pi^{1/2}}=\bvec{A}[\pi^{1/4}\otimes(\pi^{1/4})^\top]\kvec{r_0'}$ and
\begin{align}
	\ev{A,A}_s&=\ev{A,\pi^{1/2}A\pi^{1/2}}\notag\\
	&=\bvec{A}[\pi^{1/2}\otimes(\pi^{1/2})^\top]\kvec{A}\notag\\
	&=\bvec{A}[\pi^{1/4}\otimes(\pi^{1/4})^\top]\sum_{i=0}^{d^2-1}\kvec{r_i'}\bvec{r_i'}[\pi^{1/4}\otimes(\pi^{1/4})^\top]\kvec{A}\notag\\
	&=\sum_{i=0}^{d^2-1}|\bvec{A}[\pi^{1/4}\otimes(\pi^{1/4})^\top]\kvec{r_i'}|^2\notag\\
	&=|\ev{A,\mbb{1}}_s|^2+\sum_{i=1}^{d^2-1}|\bvec{A}[\pi^{1/4}\otimes(\pi^{1/4})^\top]\kvec{r_i'}|^2.
\end{align}
Using these expressions, we immediately obtain
\begin{equation}
	\max_{\ev{A,\mbb{1}}_s=0}\frac{\ev{A,\widetilde{\mca{L}}_s(A)}_s}{\ev{A,A}_s}=\max_{\ev{A,\mbb{1}}_s=0}\frac{\sum_{i=1}^{d^2-1}\lambda_i|\bvec{A}[\pi^{1/4}\otimes(\pi^{1/4})^\top]\kvec{r_i'}|^2}{\sum_{i=1}^{d^2-1}|\bvec{A}[\pi^{1/4}\otimes(\pi^{1/4})^\top]\kvec{r_i'}|^2}=\lambda_1=-g_s,
\end{equation}
where we apply the inequality $\sum_ia_ix_i/\sum_ix_i\le\max_ia_i$ for any nonnegative numbers $\{x_i\}$ and the maximum is achieved by $\kvec{A}=[\pi^{-1/4}\otimes(\pi^{-1/4})^\top]\kvec{r_1'}$ or $A=\pi^{-1/4}r_1'\pi^{-1/4}$.
\end{proof}

\section{Derivation of the quantum thermo-kinetic uncertainty relation \eqref{eq:main.result.1}}\label{app:proof.qtkur}
Here we provide a detailed derivation of the first main result \eqref{eq:main.result.1}.
Consider an auxiliary GKSL dynamics, which is perturbed from the original one by a parameter $\theta$.
Suppose that we want to estimate $\theta$ from current $\phi$, which can be a biased estimator.
According to the quantum Cram{\'e}r-Rao inequality \cite{Helstrom.1969}, the precision of this estimator is limited by the quantum Fisher information as 
\begin{equation}\label{eq:qCRB}
	\frac{\mvar[\phi]}{(\partial_{\theta}\ev{\phi})^2}\ge\frac{1}{{I}_q(\theta)},
\end{equation}
where ${I}_q(\theta)$ denotes the quantum Fisher information.
Now, we specify the auxiliary dynamics, where the Hamiltonian remains unchanged (i.e., $H_\theta=H$) and the jump operators are perturbed as
\begin{equation}\label{eq:pert.dyn}
	L_{k,\theta}=\sqrt{1+\ell_k\theta}L_k,~\ell_k=\frac{\tr(L_k\pi L_k^\dagger)-\tr(L_{k^*}\pi L_{k^*}^\dagger)}{\tr(L_k\pi L_k^\dagger)+\tr(L_{k^*}\pi L_{k^*}^\dagger)}.
\end{equation}
Inequality \eqref{eq:qCRB} holds true for various values of $\theta$.
Hereafter, we exclusively consider the $\theta\to 0$ limit.
For GKSL dynamics, the quantum Fisher information can be explicitly calculated as \cite{Gammelmark.2014.PRL}
\begin{equation}
{I}_q(0)=4\eval{\partial^2_{\theta_1\theta_2}\ln|\tr \varrho_{\vb*{\theta}}(\tau)|}_{\vb*{\theta}=\vb*{0}},
\end{equation}
where $\varrho_{\vb*{\theta}}(\tau)=e^{\mca{L}_{\vb*{\theta}}\tau}(\pi)$ is an operator evolved according to the following modified super-operator:
\begin{align}
\mca{L}_{\vb*{\theta}}(\varrho)&=-i[H,\varrho]+\sum_{k\ge 1}\sqrt{(1+\ell_k\theta_1)(1+\ell_k\theta_2)}L_k\varrho L_k^\dagger-\frac{1}{2}\sum_{k\ge 1}(1+\ell_k\theta_1)L_k^\dagger L_k\varrho-\frac{1}{2}\sum_{k\ge 1}(1+\ell_k\theta_2)\varrho L_k^\dagger L_k.
\end{align}
To calculate ${I}_q$, it is convenient to use the vectorization representation.
Using this representation, the GKSL equation can be written as $\kvec{\dot\varrho_t}=\widehat{\mca{L}}\kvec{\varrho_t}$, where the operator $\widehat{\mca{L}}$ is defined as
\begin{align}
\widehat{\mca{L}}&\coloneqq-i(H\otimes\mbb{1} - \mbb{1}\otimes H^\top)+\sum_{k\ge 1}\qty[ L_k\otimes L_k^* - \frac{1}{2} ( L_k^\dagger L_k )\otimes\mbb{1} - \frac{1}{2} \mbb{1}\otimes( L_k^\dagger L_k )^\top].
\end{align}
Here $\top$ and $*$ denote the matrix transpose and complex conjugate, respectively.
The operator $\widehat{\mca{L}}$ has eigenvalues $0=\chi_0>\Re(\chi_1)\ge\Re(\chi_2)\ge\dots$ and can be expressed in terms of the eigenvalues and its right and left eigenvectors $\{\kvec{r_i},\kvec{l_i}\}$ as
\begin{equation}
	\widehat{\mca{L}}=\sum_{i>0}\chi_i\kvec{r_i}\bvec{l_i},
\end{equation}
where $\kvec{\pi}\bvec{\mbb{1}}+\sum_{i>0}\kvec{r_i}\bvec{l_i}=\mbb{1}$.
The Moore-Penrose pseudo-inverse of $\widehat{\mca{L}}$ can be defined as
\begin{equation}
	\widehat{\mca{L}}^+\coloneqq\sum_{i>0}\chi_i^{-1}\kvec{r_i}\bvec{l_i}.\label{eq:pseudo.inv}
\end{equation}
Using $\tr \varrho_{\vb*{\theta}}(\tau)=\bvec{\mbb{1}}e^{\widehat{\mca{L}}_{\vb*{\theta}}\tau}\kvec{\pi}$ and applying the equality $\partial_ue^{\widehat{\mca{L}}_ut}=\int_0^t\dd{s}e^{\widehat{\mca{L}}_u(t-s)}\partial_u\widehat{\mca{L}}_ue^{\widehat{\mca{L}}_us}$, the quantum Fisher information can be calculated as follows:
\begin{align}
	{I}_q(0)&=4\eval{\qty[\partial_{\theta_1\theta_2}^2\bvec{\mbb{1}}e^{\widehat{\mca{L}}_{\vb*{\theta}}\tau}\kvec{\pi}-\partial_{\theta_1}\bvec{\mbb{1}}e^{\widehat{\mca{L}}_{\vb*{\theta}}\tau}\kvec{\pi}\partial_{\theta_2}\bvec{\mbb{1}}e^{\widehat{\mca{L}}_{\vb*{\theta}}\tau}\kvec{\pi}]}_{\vb*{\theta}=\vb*{0}}\notag\\
	&=4\eval{\partial_{\theta_2}\bvec{\mbb{1}}\int_0^\tau\dd{t}e^{\widehat{\mca{L}}_{\vb*{\theta}}(\tau-t)}\partial_{\theta_1}\widehat{\mca{L}}_{\vb*{\theta}}e^{\widehat{\mca{L}}_{\vb*{\theta}}t}\kvec{\pi}}_{\vb*{\theta}=\vb*{0}}-4\eval{\bvec{\mbb{1}}\int_0^\tau\dd{t}e^{\widehat{\mca{L}}_{\vb*{\theta}}(\tau-t)}\partial_{\theta_1}\widehat{\mca{L}}_{\vb*{\theta}}e^{\widehat{\mca{L}}_{\vb*{\theta}}t}\kvec{\pi}\bvec{\mbb{1}}\int_0^\tau\dd{t}e^{\widehat{\mca{L}}_{\vb*{\theta}}(\tau-t)}\partial_{\theta_2}\widehat{\mca{L}}_{\vb*{\theta}}e^{\widehat{\mca{L}}_{\vb*{\theta}}t}\kvec{\pi}}_{\vb*{\theta}=\vb*{0}}\notag\\
	&=-4\eval{\bvec{\mbb{1}}\int_0^\tau\dd{t}e^{\widehat{\mca{L}}_{\vb*{\theta}}(\tau-t)}\partial_{\theta_1}\widehat{\mca{L}}_{\vb*{\theta}}e^{\widehat{\mca{L}}_{\vb*{\theta}}t}\kvec{\pi}\bvec{\mbb{1}}\int_0^\tau\dd{t}e^{\widehat{\mca{L}}_{\vb*{\theta}}(\tau-t)}\partial_{\theta_2}\widehat{\mca{L}}_{\vb*{\theta}}e^{\widehat{\mca{L}}_{\vb*{\theta}}t}\kvec{\pi}}_{\vb*{\theta}=\vb*{0}}\notag\\
	&+4\eval{\bvec{\mbb{1}}\int_0^\tau\dd{t}\int_0^{\tau-t}\dd{s}e^{\widehat{\mca{L}}_{\vb*{\theta}}(\tau-t-s)}\partial_{\theta_2}\widehat{\mca{L}}_{\vb*{\theta}}e^{\widehat{\mca{L}}_{\vb*{\theta}}s}\partial_{\theta_1}\widehat{\mca{L}}_{\vb*{\theta}}e^{\widehat{\mca{L}}_{\vb*{\theta}}t}\kvec{\pi}}_{\vb*{\theta}=\vb*{0}}\notag\\
	&+4\eval{\bvec{\mbb{1}}\int_0^\tau\dd{t}\int_0^{t}\dd{s}e^{\widehat{\mca{L}}_{\vb*{\theta}}(\tau-t)}\partial_{\theta_1}\widehat{\mca{L}}_{\vb*{\theta}}e^{\widehat{\mca{L}}_{\vb*{\theta}}(t-s)}\partial_{\theta_2}\widehat{\mca{L}}_{\vb*{\theta}}e^{\widehat{\mca{L}}_{\vb*{\theta}}s}\kvec{\pi}}_{\vb*{\theta}=\vb*{0}}\notag\\
	&+4\eval{\bvec{\mbb{1}}\int_0^\tau\dd{t}e^{\widehat{\mca{L}}_{\vb*{\theta}}(\tau-t)}\partial_{\theta_1\theta_2}^2\widehat{\mca{L}}_{\vb*{\theta}}e^{\widehat{\mca{L}}_{\vb*{\theta}}t}\kvec{\pi}}_{\vb*{\theta}=\vb*{0}}\notag\\
	&=4\qty[-\tau^2\bvec{\mbb{1}}\widehat{\mca{F}}_1\kvec{\pi}\bvec{\mbb{1}}\widehat{\mca{F}}_2\kvec{\pi}+\bvec{\mbb{1}}\widehat{\mca{F}}_2\widehat{\mca{E}}_{\tau}\widehat{\mca{F}}_1\kvec{\pi}+\bvec{\mbb{1}}\widehat{\mca{F}}_1\widehat{\mca{E}}_{\tau}\widehat{\mca{F}}_2\kvec{\pi}+\tau\bvec{\mbb{1}}\eval{\partial_{\theta_1\theta_2}^2\widehat{\mca{L}}_{\vb*{\theta}}\kvec{\pi}}_{\vb*{\theta}=\vb*{0}}],\label{eq:qFI.tmp1}
\end{align}
where we use the facts $\bvec{\mbb{1}}e^{\widehat{\mca{L}}t}=\bvec{\mbb{1}}$ and $e^{\widehat{\mca{L}}t}\kvec{\pi}=\kvec{\pi}$ for any $t\ge 0$, and the operators $\widehat{\mca{E}}_{\tau}$, $\widehat{\mca{F}}_1$, and $\widehat{\mca{F}}_2$ are given by
\begin{align}
	\widehat{\mca{E}}_{\tau}&\coloneqq\int_0^\tau\dd{t}\int_0^{t}\dd{s}e^{\widehat{\mca{L}}s},\notag\\
	\widehat{\mca{F}}_1&\coloneqq\eval{\partial_{\theta_1}\widehat{\mca{L}}_{\vb*{\theta}}}_{\vb*{\theta}=\vb*{0}}=\frac{1}{2}\sum_{k\ge 1}\ell_k\qty[ L_k\otimes L_k^* - (L_k^\dagger L_k )\otimes\mbb{1}],\notag\\
	\widehat{\mca{F}}_2&\coloneqq\eval{\partial_{\theta_2}\widehat{\mca{L}}_{\vb*{\theta}}}_{\vb*{\theta}=\vb*{0}}=\frac{1}{2}\sum_{k\ge 1}\ell_k\qty[ L_k\otimes L_k^* - \mbb{1}\otimes(L_k^\dagger L_k )^\top].
\end{align}
Since $\bvec{\mbb{1}}\widehat{\mca{F}}_1\kvec{A}=\bvec{\mbb{1}}\widehat{\mca{F}}_2\kvec{A}=0$ for any operator $A$, we readily get
\begin{equation}
{I}_q(0)=4\tau\bvec{\mbb{1}}\eval{\partial_{\theta_1\theta_2}^2\widehat{\mca{L}}_{\vb*{\theta}}\kvec{\pi}}_{\vb*{\theta}=\vb*{0}}=\tau\bvec{\mbb{1}}\sum_{k\ge 1}\ell_k^2[L_k\otimes L_k^*]\kvec{\pi}=\tau\sum_{k\ge 1}\ell_k^2\tr(L_k\pi L_k^\dagger).
\end{equation}
Let $\pi=\sum_n\pi_n\dyad{n}$ be the spectral decomposition of $\pi$.
For convenience, we define $w_{mn}^k\coloneqq|\mel{m}{L_k}{n}|^2$, $\sigma_{mn}^k\coloneqq(w_{mn}^k\pi_n-w_{nm}^{k^*}\pi_m)\ln(w_{mn}^k\pi_n/w_{nm}^{k^*}\pi_m)$, and $a_{mn}^k\coloneqq(w_{mn}^k\pi_n+w_{nm}^{k^*}\pi_m)$.
Note that $\sigma=(1/2)\sum_{k\ge 1,m,n}\sigma_{mn}^k$ and $a=(1/2)\sum_{k\ge 1,m,n}a_{mn}^k$ \cite{Vu.2023.PRX}.
We can upper bound the quantum Fisher information by the rates of irreversible entropy production and dynamical activity as follows:
\begin{align}
	{I}_q(0)&=\frac{1}{2}\tau\sum_{k\ge 1}\ell_k^2[\tr(L_k\pi L_k^\dagger)+\tr(L_{k^*}\pi L_{k^*}^\dagger)]\notag\\
	&=\frac{1}{2}\tau\sum_{k\ge 1}\frac{\qty[\tr(L_k\pi L_k^\dagger)-\tr(L_{k^*}\pi L_{k^*}^\dagger)]^2}{\tr(L_k\pi L_k^\dagger)+\tr(L_{k^*}\pi L_{k^*}^\dagger)}\notag\\
	&=\frac{1}{2}\tau\sum_{k\ge 1}\frac{(\sum_{m,n}[w_{mn}^k\pi_n-w_{nm}^{k^*}\pi_m])^2}{\sum_{m,n}[w_{mn}^k\pi_n+w_{nm}^{k^*}\pi_m]}\notag\\
	&\le\frac{1}{2}\tau\sum_{k\ge 1,m,n}\frac{(w_{mn}^k\pi_n-w_{nm}^{k^*}\pi_m)^2}{w_{mn}^k\pi_n+w_{nm}^{k^*}\pi_m}\notag\\
	&=\tau\sum_{k\ge 1,m,n}\frac{(\sigma_{mn}^k)^2}{8a_{mn}^k}\Phi\qty(\frac{\sigma_{mn}^k}{2a_{mn}^k})^{-2}\notag\\
	&\le\tau\frac{\sigma^2}{4a}\Phi\qty(\frac{\sigma}{2a})^{-2}\notag\\
	&\le\tau\min\qty(\frac{\sigma}{2},a).\label{eq:qFI.ub}
\end{align}
Here we use the fact that $(x^2/y)\Phi(x/y)^{-2}$ is a concave function and apply Jensen's inequality to obtain the sixth line.

Next, we calculate the term $\partial_\theta\ev{\phi}$.
For $\theta\ll 1$, the density operator $\varrho_{t,\theta}$ in the auxiliary dynamics \eqref{eq:pert.dyn} can be expanded in terms of $\theta$ as $\varrho_{t,\theta}=\pi+\theta\varphi_t+O(\theta^2)$.
Substituting this perturbative expression to the GKSL equation and collecting the terms in the first order of $\theta$, we obtain the differential equation that describes the time evolution of the operator $\varphi_t$,
\begin{equation}\label{eq:phi.evo}
\dot\varphi_t=\mca{L}(\varphi_t)+\sum_{k\ge 1}\ell_k(L_k\pi L_k^\dagger-\{L_k^\dagger L_k,\pi\}/2),
\end{equation}
where the initial condition is given by $\varphi_0=\mbb{0}$.
It can be easily seen that the operator $\varphi_t$ is always traceless.
Noting that $c_k=-c_{k^*}$ and $c_k\ell_k=c_{k^*}\ell_{k^*}$, the partial derivative of the current average in the auxiliary dynamics with respect to $\theta$ can be calculated as
\begin{align}
\eval{\partial_{\theta}\expval{\phi}}_{\theta=0}&=\eval{\partial_{\theta}\qty[\int_0^\tau\dd{t}\sum_{k\ge 1}c_k(1+\ell_k\theta)\tr{L_k(\pi+\theta\varphi_t)L_k^\dagger}+O(\theta^2)]}_{\theta=0}\notag\\
&=\int_0^\tau\dd{t}\sum_{k\ge 1}c_k\ell_k\tr(L_k\pi L_k^\dagger)+\int_0^\tau\dd{t}\sum_{k\ge 1}c_k\tr(L_k\varphi_tL_k^\dagger)\notag\\
&=\frac{1}{2}\int_0^\tau\dd{t}\sum_{k\ge 1}c_k\ell_k\qty[\tr(L_k\pi L_k^\dagger)+\tr(L_{k^*}\pi L_{k^*}^\dagger)]+\int_0^\tau\dd{t}\sum_{k\ge 1}c_k\tr(L_k\varphi_tL_k^\dagger)\notag\\
&=\frac{1}{2}\int_0^\tau\dd{t}\sum_{k\ge 1}c_k\qty[\tr(L_k\pi L_k^\dagger)-\tr(L_{k^*}\pi L_{k^*}^\dagger)]+\int_0^\tau\dd{t}\sum_{k\ge 1}c_k\tr(L_k\varphi_tL_k^\dagger)\notag\\
&=\tau\sum_{k\ge 1}c_k\tr(L_k\pi L_k^\dagger)+\int_0^\tau\dd{t}\bvec{\mbb{1}}\widehat{\mca{C}}\kvec{\varphi_t}\notag\\
&=\expval{\phi}+\expval{\phi}_\varphi,\label{eq:par.avg.TUR}
\end{align}
where we define $\expval{\phi}_\varphi\coloneqq\int_0^\tau\dd{t}\bvec{\mbb{1}}\widehat{\mca{C}}\kvec{\varphi_t}$ and $\widehat{\mca{C}}\coloneqq\sum_{k\ge 1}c_kL_k\otimes L_k^*$.
Defining the following quantum correction of the current average:
\begin{equation}\label{eq:cur.avg.qcor}
	\delta_\phi\coloneqq\frac{\ev{\phi}_\varphi}{\ev{\phi}},
\end{equation}
we readily obtain the quantum TKUR \eqref{eq:main.result.1} from Eqs.~\eqref{eq:qCRB}, \eqref{eq:qFI.ub}, and \eqref{eq:par.avg.TUR}, 
\begin{equation}
\frac{F_\phi}{(1+\delta_\phi)^2}\ge\frac{4a}{\sigma^2}\Phi\qty(\frac{\sigma}{2a})^2.
\end{equation}
The asymptotic long-time value of $\delta_\phi$ can be analytically calculated.
In the long-time limit, $\varphi_t$ converges to a stationary operator $\varphi_{\rm ss}$, which satisfies the following equation [cf.~\eqref{eq:phi.evo}]:
\begin{equation}
	\widehat{\mca{L}}\kvec{\varphi_{\rm ss}}=-\widehat{\mca{D}}_\ell\kvec{\pi},~\text{where}~\widehat{\mca{D}}_\ell\coloneqq\sum_{k\ge 1}\ell_k\qty[ L_k\otimes L_k^* - \frac{1}{2} ( L_k^\dagger L_k )\otimes\mbb{1} - \frac{1}{2} \mbb{1}\otimes( L_k^\dagger L_k )^\top].\label{eq:phiss.tmp}
\end{equation}
Multiplying the pseudo inverse $\widehat{\mca{L}}^+$ to the both sides of Eq.~\eqref{eq:phiss.tmp} and noting that $\tr\varphi_{\rm ss}=0$, we obtain $\kvec{\varphi_{\rm ss}}=-\widehat{\mca{L}}^+\widehat{\mca{D}}_\ell\kvec{\pi}$.
Thus, the term $\ev{\phi}_\varphi$ can be calculated as
\begin{equation}
	\lim_{\tau\to\infty}\tau^{-1}\ev{\phi}_\varphi=-\bvec{\mbb{1}}\widehat{\mca{C}}\widehat{\mca{L}}^+\widehat{\mca{D}}_\ell\kvec{\pi}.
\end{equation}
Consequently, $\delta_\phi$ can be explicitly expressed as
\begin{equation}\label{eq:deltaJ.def}
	\delta_\phi= -\frac{\bvec{\mbb{1}}\widehat{\mca{C}}\widehat{\mca{L}}^+\widehat{\mca{D}}_\ell\kvec{\pi}}{\bvec{\mbb{1}}\widehat{\mca{C}}\kvec{\pi}}.
\end{equation}
In the classical limit (i.e., $H=\sum_n\epsilon_n\dyad{n}$ and $L_k=\sqrt{\gamma_{mn}}\dyad{m}{n}$), we can calculate as follows:
\begin{align}
	\sum_{k\ge 1}\ell_k(L_k\pi L_k^\dagger-\{L_k^\dagger L_k,\pi\}/2)&=\sum_{m\neq n}\frac{\gamma_{mn}\pi_n-\gamma_{nm}\pi_m}{\gamma_{mn}\pi_n+\gamma_{nm}\pi_m}\qty(\gamma_{mn}\pi_n\dyad{m}-\gamma_{mn}\pi_n\dyad{n})\notag\\
	&=\sum_{m}\sum_{n(\neq m)}\frac{\gamma_{mn}\pi_n-\gamma_{nm}\pi_m}{\gamma_{mn}\pi_n+\gamma_{nm}\pi_m}\qty(\gamma_{mn}\pi_n+\gamma_{nm}\pi_m)\dyad{m}\notag\\
	&=\sum_{m}\dyad{m}\sum_{n(\neq m)}(\gamma_{mn}\pi_n-\gamma_{nm}\pi_m)\notag\\
	&=\mbb{0}.
\end{align}
This means that Eq.~\eqref{eq:phi.evo} reduces to $\dot\varphi_t=\mca{L}(\varphi_t)$ with the initial condition $\varphi_0=\mbb{0}$.
Consequently, $\varphi_t=\mbb{0}$ for all $t\ge 0$.
Therefore, we can show $\delta_\phi=0$ in the classical limit as
\begin{equation}
	\delta_\phi=\ev{\phi}^{-1}\int_0^\tau\dd{t}\bvec{\mbb{1}}\widehat{\mca{C}}\kvec{\varphi_t}=0.\label{eq:deltaJ.cl}
\end{equation}

\subsection{Analogous relation for quantum diffusion unraveling}\label{app:qTKUR.qdu}
Here we demonstrate that an analogous relation can be derived for quantum diffusion unraveling using the same approach.
To this end, we first briefly outline the method of quantum diffusion unraveling \cite{Breuer.2002} (see also Ref.~\cite{Landi.2024.PRXQ} for a review). Unlike quantum jump unraveling, where the quantum state evolves through discontinuous jumps, quantum diffusion unraveling describes a continuous stochastic evolution due to a diffusion-like noise process.
For simplicity, we consider real coherent fields, while generalization to complex cases is straightforward. In this case, the pure state $\ket{\psi_t}$ follows the stochastic differential equation:
\begin{equation}
	d\ket{\psi_t}=\qty(-iH_{\rm eff}+\sum_{k\ge 1}\qty[\ev{L_k^\dagger}_tL_k-|\ev{L_k}_t|^2/2])\ket{\psi_t}dt+\sum_{k\ge 1}\qty(L_k-\ev{L_k}_t)\ket{\psi_t}dW_{k,t}.
\end{equation}
Here, $\{dW_{k,t}\}$ are the independent Wiener increments satisfying $\mbb{E}[dW_{k,t}]=0$ and $\mbb{E}[dW_{k,t}dW_{k',t}]=\delta_{kk'}dt$.
The current observable $\phi$ along a stochastic trajectory is given by
\begin{equation}
	\phi=\int_0^\tau\sum_{k\ge 1}c_k\qty(\ev{L_k+L_k^\dagger}_t\dd{t}+dW_{k,t}).
\end{equation}
Its ensemble average can be calculated as
\begin{equation}
	\ev{\phi}=\int_0^\tau\dd{t}\sum_{k\ge 1}c_k\tr[(L_k+L_k^\dagger)\varrho_t]=\tau\sum_{k\ge 1}c_k\tr[(L_k+L_k^\dagger)\pi].
\end{equation}
The quantum Cram{\'e}r-Rao inequality \eqref{eq:qCRB} remains valid even for quantum diffusion unraveling. The key distinction from the case of quantum jump unraveling lies in the partial derivative of the current observable. In this case, it is evaluated as follows: 
\begin{align}
\eval{\partial_{\theta}\expval{\phi}}_{\theta=0}&=\eval{\partial_{\theta}\qty[\int_0^\tau\dd{t}\sum_{k\ge 1}c_k(1+\ell_k\theta/2)\tr[(L_k+L_k^\dagger)(\pi+\theta\varphi_t)]+O(\theta^2)]}_{\theta=0}\notag\\
&=\frac{\tau}{2}\sum_{k\ge 1}c_k\ell_k\tr[(L_k+L_k^\dagger)\pi]+\int_0^\tau\dd{t}\sum_{k\ge 1}c_k\tr[(L_k+L_k^\dagger)\varphi_t]\notag\\
&=\expval{\phi}+\expval{\phi}_*+\expval{\phi}_\varphi,\label{eq:qdu.par.avg.TUR}
\end{align}
where $\expval{\phi}_*\coloneqq\tau\sum_{k\ge 1}c_k(\ell_k/2-1)\tr[(L_k+L_k^\dagger)\pi]$.
Defining the quantum correction term $\delta_\phi'\coloneqq(\expval{\phi}_*+\expval{\phi}_\varphi)/\expval{\phi}$, we obtain the following quantum TKUR for quantum diffusion unraveling:
\begin{equation}
	\frac{F_\phi}{(1+\delta_\phi')^2}\ge\frac{4a}{\sigma^2}\Phi\qty(\frac{\sigma}{2a})^2.\label{eq:qTKUR.qdu}
\end{equation}
This relation retains the same structure as Eq.~\eqref{eq:main.result.1} derived for quantum jump unraveling, differing only in the quantum contribution $\delta_\phi'$.

\section{Derivation of the quantum inverse uncertainty relation \eqref{eq:main.result.2}}\label{app:proof.qiur}
Here we provide a detailed derivation of the second main result \eqref{eq:main.result.2}.
For notational convenience, we denote $\partial_xf_x$ as $f'_x$ and $\partial_x^2f_x$ as $f''_x$, respectively.
We first calculate the variance of observable $\phi$.
Applying the equality
\begin{equation}
	\partial_ue^{\mca{L}_ut}=\int_0^t\dd{s}e^{\mca{L}_u(t-s)}\mca{L}_u'e^{\mca{L}_us},
\end{equation}
the second moment of observable $\phi$ can be calculated from Eq.~\eqref{eq:C.nth.moments} as
\begin{align}
	\ev{\phi^2}&=-\partial_u^2G_\tau(u)\big|_{u=0}\notag\\
	&=-\int_0^\tau\dd{t}\tr\mca{L}_0''(\pi)-2\int_0^\tau\dd{t}\int_0^{t}\dd{s}\tr{\mca{L}_0'e^{\mca{L}s}\mca{L}_0'(\pi)}.\label{eq:res2.tmp1}
\end{align}
Here, we use the facts that $\mca{L}_0=\mca{L}$, $e^{\mca{L}t}(\pi)=\pi$, and $\tr e^{\mca{L}t}(A)=\tr A$ for any operator $A$.
We now individually calculate the terms in Eq.~\eqref{eq:res2.tmp1}.
For convenience, we define the following self-adjoint operators: $J_1\coloneqq\sum_{k\ge 1}c_kL_k^\dagger L_k$, $J_2\coloneqq \sum_{k\ge 1}c_k^2L_k^\dagger L_k$, and $J_{\pi} \coloneqq\sum_{k\ge 1}c_kL_k\pi L_k^\dagger$.
Evidently, $\tau^{-1}\ev{\phi}=\ev{J_1,\pi}=\tr(J_\pi)$.
Since $\tr\mca{L}_0''(\pi)=-\sum_{k\ge 1}c_k^2\tr(L_k\pi L_k^\dagger)=-\ev{J_2,\pi}$, the first term in Eq.~\eqref{eq:res2.tmp1} becomes
\begin{equation}
	-\int_0^\tau\dd{t}\tr\mca{L}_0''(\pi)=\tau\ev{J_2,\pi}.\label{eq:res2.tmp5}
\end{equation}
In addition, note that $\mca{L}_0'(\circ)=i\sum_{k\ge 1}c_kL_k\circ L_k^\dagger$ and 
\begin{align}
	\tr\mca{L}_0'(\circ)&=i\sum_{k\ge 1}c_k\tr(L_k \circ L_k^\dagger)\notag\\
	&=i\sum_{k\ge 1}c_k\tr(L_k^\dagger L_k\circ)\notag\\
	&=i\ev{J_1,\circ}.
\end{align}
The second term in Eq.~\eqref{eq:res2.tmp1} can be calculated as
\begin{align}
	-2\int_0^\tau\dd{t}\int_0^{t}\dd{s}\tr{\mca{L}_0'e^{\mca{L}s}\mca{L}_0'(\pi)}&=-2i\int_0^\tau\dd{t}\int_0^{t}\dd{s}\ev{J_1,e^{\mca{L}s}\mca{L}_0'(\pi)}\notag\\
	&=-2i\int_0^\tau\dd{t}\int_0^{t}\dd{s}\ev{e^{\widetilde{\mca{L}}s}(J_1),\mca{L}_0'(\pi)}\notag\\
	&=2\int_0^\tau\dd{t}\int_0^{t}\dd{s}\ev{e^{\widetilde{\mca{L}}s}(J_1),J_\pi}\notag\\
	&=2\int_0^\tau\dd{t}(\tau-t)\ev{e^{\widetilde{\mca{L}}t}(J_1),J_\pi}\notag\\
	&=2\int_0^\tau\dd{t}(\tau-t)\qty[\ev{e^{\widetilde{\mca{L}}t}(J_1-\ev{J_1,\pi}\mbb{1}),J_\pi}+\ev{J_1,\pi}^2]\notag\\
	&=\ev{\phi}^2+2\int_0^\tau\dd{t}(\tau-t)\ev{e^{\widetilde{\mca{L}}t}(\bar{J}_1),J_\pi},\label{eq:res2.tmp2}
\end{align}
where we define $\bar{J}_1\coloneqq J_1-\ev{J_1,\pi}\mbb{1}$, which is a self-adjoint operator satisfying $\ev{\bar{J}_1,\pi}=0$.
For each operator $A$, we define the corresponding operator $A_t$ that evolves in the Heisenberg picture as $A_t\coloneqq e^{\widetilde{\mca{L}}t}(A)$.
Applying the Cauchy-Schwarz inequality with respect to the inner product $\ev{\cdot,\cdot}_s$, the last term in Eq.~\eqref{eq:res2.tmp2} can be upper bounded as follows:
\begin{align}
	\ev{e^{\widetilde{\mca{L}}t}(\bar{J}_1),J_\pi}&=\ev{\bar{J}_{1,t},J_\pi-\alpha\pi}\notag\\
	&=\ev{\bar{J}_{1,t},\pi^{-s}(J_\pi-\alpha\pi)\pi^{s-1}}_s\notag\\
	&\le\|\bar{J}_{1,t}\|_s\|\pi^{-s}(J_\pi-\alpha\pi)\pi^{s-1}\|_s\notag\\
	&=\|\bar{J}_{1,t}\|_s\sqrt{\tr{\pi^{s-1}(J_\pi-\alpha\pi)\pi^{-s}\pi^{s}\pi^{-s}(J_\pi-\alpha\pi)\pi^{s-1}\pi^{1-s}}}\notag\\
	&=\|\bar{J}_{1,t}\|_s\sqrt{\tr{\pi^{s-1}(J_\pi-\alpha\pi)\pi^{-s}(J_\pi-\alpha\pi)}}.\label{eq:res2.tmp3}
\end{align}
Here, $\alpha$ is an arbitrary real number.
The last term in Eq.~\eqref{eq:res2.tmp3} can be expressed as a quadratic function of $\alpha$ as
\begin{align}
	\tr{\pi^{s-1}(J_\pi-\alpha\pi)\pi^{-s}(J_\pi-\alpha\pi)}&=\alpha^2-2\alpha\ev{J_1,\pi}+\tr(\pi^{s-1}J_\pi\pi^{-s}J_\pi).
\end{align}
This quantity can be minimized by $\alpha=\ev{J_1,\pi}$ and the minimum is given by
\begin{align}
	\tr(\pi^{s-1}J_\pi\pi^{-s}J_\pi)-\ev{J_1,\pi}^2&=\|\pi^{-s}J_\pi\pi^{s-1}\|_s^2-\ev{J_1,\pi}^2\notag\\
	&=\|\pi^{-s}J_\pi\pi^{s-1}\|_s^2-\ev{\pi^{-s}J_\pi\pi^{s-1},\pi}^2\notag\\
	&=\|\pi^{-s}J_\pi\pi^{s-1}-\ev{J_1,\pi}\mbb{1}\|_s^2.
\end{align}
Here, we use the fact that the following equality holds true for operator $A=\pi^{-s}J_\pi\pi^{s-1}$:
\begin{align}
	\|A-\ev{A,\pi}\mbb{1}\|_s^2&=\tr{[A^\dagger-\ev{A,\pi}\mbb{1}]\pi^s[A-\ev{A,\pi}\mbb{1}]\pi^{1-s}}\notag\\
	&=\|A\|_s^2-\ev{A,\pi}^2.
\end{align}
Consequently, the following upper bound can be attained from Eq.~\eqref{eq:res2.tmp3}:
\begin{align}
	\ev{e^{\widetilde{\mca{L}}t}(\bar{J}_1),J_\pi}&\le\|\bar{J}_{1,t}\|_s\|\pi^{-s}J_\pi\pi^{s-1}-\ev{J_1,\pi}\mbb{1}\|_s.
\end{align}
Next, we evaluate $\|\bar{J}_{1,t}\|_s$.
Note that $\ev{\bar{J}_1,\pi}=0$, which implies $\ev{\bar{J}_{1,t},\mbb{1}}_s=\ev{\bar{J}_{1,t},\pi}=\ev{\bar{J}_1,e^{\mca{L}t}(\pi)}=0$.
Following the approach in Ref.~\cite{Mori.2023.PRL} and applying Lemma \ref{lem:sym.Lio}, we can show that
\begin{align}
	\frac{d}{dt}\ev{\bar{J}_{1,t},\bar{J}_{1,t}}_s&=\ev{\widetilde{\mca{L}}(\bar{J}_{1,t}),\bar{J}_{1,t}}_s+\ev{\bar{J}_{1,t},\widetilde{\mca{L}}(\bar{J}_{1,t})}_s\notag\\
	&=\ev{\bar{J}_{1,t},\widetilde{\mca{L}}^*(\bar{J}_{1,t})}_s+\ev{\bar{J}_{1,t},\widetilde{\mca{L}}(\bar{J}_{1,t})}_s\notag\\
	&=2\ev{\bar{J}_{1,t},\widetilde{\mca{L}}_s(\bar{J}_{1,t})}_s\notag\\
	&=2\ev{\bar{J}_{1,t},\bar{J}_{1,t}}_s\frac{\ev{\bar{J}_{1,t},\widetilde{\mca{L}}_s(\bar{J}_{1,t})}_s}{\ev{\bar{J}_{1,t},\bar{J}_{1,t}}_s}\notag\\
	&\le -2g_s\ev{\bar{J}_{1,t},\bar{J}_{1,t}}_s,
\end{align}
which consequently derives $\|\bar{J}_{1,t}\|_s\le e^{-g_st}\|\bar{J}_1\|_s$.
Therefore, we arrive at the following inequality:
\begin{equation}
	\ev{e^{\widetilde{\mca{L}}t}(\bar{J}_1),J_\pi}\le e^{-g_st}\|\bar{J}_1\|_s\|\pi^{-s}J_\pi\pi^{s-1}-\ev{J_1,\pi}\mbb{1}\|_s=e^{-g_st}\|J_1-\ev{J_1,\pi}\mbb{1}\|_s\|\pi^{-s}J_\pi\pi^{s-1}-\ev{J_1,\pi}\mbb{1}\|_s.
\end{equation}
Using this inequality and taking the time integration in Eq.~\eqref{eq:res2.tmp2}, the second term in Eq.~\eqref{eq:res2.tmp1} can be upper bounded as
\begin{equation}
	-2\int_0^\tau\dd{t}\int_0^{t}\dd{s}\tr{\mca{L}_0'e^{\mca{L}s}\mca{L}_0'(\pi)}\le\ev{\phi}^2+\frac{2(e^{-g_s\tau}+g_s\tau-1)}{g_s^2}\|J_1-\ev{J_1,\pi}\mbb{1}\|_s\|\pi^{-s}J_\pi\pi^{s-1}-\ev{J_1,\pi}\mbb{1}\|_s.\label{eq:res2.tmp4}
\end{equation}
Finally, by combining Eqs.~\eqref{eq:res2.tmp1}, \eqref{eq:res2.tmp5}, and \eqref{eq:res2.tmp4}, we obtain the following upper bound on the variance:
\begin{equation}
	\mvar[\phi]\le \tau\ev{J_2,\pi}+\frac{2(e^{-g_s\tau}+g_s\tau-1)}{g_s^2}\|J_1-\ev{J_1,\pi}\mbb{1}\|_s\|\pi^{-s}J_\pi\pi^{s-1}-\ev{J_1,\pi}\mbb{1}\|_s.\label{eq:res2.tmp6}
\end{equation}
By rearranging Eq.~\eqref{eq:res2.tmp6} and noting that $e^{-g_s\tau}-1\le 0$, we achieve the desired relation \eqref{eq:main.result.2},
\begin{equation}
	F_\phi\le\frac{\ev{J_2,\pi}}{\ev{J_1,\pi}^2}\qty(1+\frac{2}{g_s}\frac{\|J_1-\ev{J_1,\pi}\mbb{1}\|_s\|\pi^{-s}J_\pi\pi^{s-1}-\ev{J_1,\pi}\mbb{1}\|_s}{\ev{J_2,\pi}}).
\end{equation}

\subsection{Analogous relation for quantum diffusion unraveling}\label{app:qiur.qdu}
Here we show that a similar bound can be derived for quantum diffusion unraveling.
We employ the same approach by computing the variance of the observable $\phi$ using its moment generating function $G_\tau(u)=\tr e^{\mca{L}_u\tau}(\pi)$.
In the case of quantum diffusion unraveling, the tilted super-operator $\mca{L}_u$ is explicitly given by \cite{Landi.2024.PRXQ}
\begin{equation}
	\mca{L}_u(\varrho)=\mca{L}(\varrho)+iu\sum_{k\ge 1}c_k(L_k\varrho+\varrho L_k^\dagger)-\frac{u^2}{2}\varrho\sum_{k\ge 1}c_k^2.
\end{equation}
Analogous to the case of quantum jump unraveling, we introduce the self-adjoint operators $J_{1,d}\coloneqq\sum_{k\ge 1}c_k(L_k+L_k^\dagger)$, $J_{2,d}\coloneqq(\sum_{k\ge 1}c_k^2)\mds{1}$, and $J_{\pi,d}\coloneqq\sum_{k\ge 1}c_k(L_k\pi+\pi L_k^\dagger)$ for notational convenience.
The second moment of $\phi$ can be computed as in Eq.~\eqref{eq:res2.tmp1}, where the first term is evaluated as
\begin{align}
	-\int_0^\tau\dd{t}\tr\mca{L}_0''(\pi)&=\tau\sum_{k\ge 1}c_k^2=\tau\ev{J_{2,d},\pi}.
\end{align}
Noting the relations $\mca{L}_0'(\varrho)=i\sum_{k\ge 1}c_k(L_k\varrho+\varrho L_k^\dagger)$, $\tr \mca{L}_0'(\circ)=i\ev{J_{1,d},\circ}$, and $\ev{J_{1,d},\pi}=\tau^{-1}\ev{\phi}$, the second term can be similarly calculated as
\begin{align}
	-2\int_0^\tau\dd{t}\int_0^{t}\dd{s}\tr{\mca{L}_0'e^{\mca{L}s}\mca{L}_0'(\pi)}&=-2i\int_0^\tau\dd{t}\int_0^{t}\dd{s}\ev{J_{1,d},e^{\mca{L}s}\mca{L}_0'(\pi)}\notag\\
	&=2\int_0^\tau\dd{t}(\tau-t)\ev{e^{\widetilde{\mca{L}}t}(J_{1,d}),J_{\pi,d}}\notag\\
	&=2\int_0^\tau\dd{t}(\tau-t)\qty[\ev{e^{\widetilde{\mca{L}}t}(J_{1,d}-\ev{J_{1,d},\pi}\mbb{1}),J_{\pi,d}}+\ev{J_{1,d},\pi}^2]\notag\\
	&=\ev{\phi}^2+2\int_0^\tau\dd{t}(\tau-t)\ev{e^{\widetilde{\mca{L}}t}(\bar{J}_{1,d}),J_{\pi,d}},
\end{align}
where we define $\bar{J}_{1,d}\coloneqq J_{1,d}-\ev{J_{1,d},\pi}\mbb{1}$, which is a self-adjoint operator satisfying $\ev{\bar{J}_{1,d},\pi}=0$.
Following the same procedure as in Eqs.~\eqref{eq:res2.tmp3}--\eqref{eq:res2.tmp6}, we obtain the quantum inverse uncertainty relation for quantum diffusion unraveling,
\begin{equation}
	F_\phi\le\frac{\ev{J_{2,d},\pi}}{\ev{J_{1,d},\pi}^2}\qty(1+\frac{2}{g_s}\frac{\|J_{1,d}-\ev{J_{1,d},\pi}\mbb{1}\|_s\|\pi^{-s}J_{\pi,d}\pi^{s-1}-\ev{J_{1,d},\pi}\mbb{1}\|_s}{\ev{J_{2,d},\pi}}).
\end{equation}
This inequality has the same structure as relation \eqref{eq:main.result.2}, where $J_1$, $J_2$, and $J_{\pi}$ are replaced by their counterparts $J_{1,d}$, $J_{2,d}$, and $J_{\pi,d}$.

\section{Derivation of the quantum response kinetic uncertainty relation \eqref{eq:main.result.3}}\label{app:proof.qrkur}
Here we provide a detailed derivation of the third main result \eqref{eq:main.result.3}.
We employ an approach similar to that in Appendix \ref{app:proof.qtkur}.
Consider an auxiliary GKSL dynamics, where the Hamiltonian remains unchanged (i.e., $H_\theta=H$) and the jump operators are perturbed by a parameter $\theta$ as
\begin{equation}\label{eq:res.per.dyn}
	L_{k,\theta}=e^{\omega_k(\epsilon[1+z_k\theta])/2}V_k,~z_k=\frac{\sign[d_{\omega_k(\epsilon)}\ev{f(\vb*{\phi})}]}{d_\epsilon\omega_k(\epsilon)},
\end{equation}
where $\sign(x)=1$ for $x\ge 0$ and $\sign(x)=-1$ otherwise.
Note that $x\sign(x)=|x|$. 
According to the quantum Cram{\'e}r-Rao inequality, the following inequality holds true:
\begin{equation}\label{eq:q.CR.bound.p}
	\frac{\mvar[f(\vb*{\phi})]}{(\partial_{\theta}\ev{f(\vb*{\phi})})^2}\ge\frac{1}{{I}_q(\theta)},
\end{equation}
where ${I}_q(\theta)$ is the quantum Fisher information.
Again, we consider the $\theta\to 0$ limit.
For the entire time regime, the quantum Fisher information can be explicitly calculated as
\begin{equation}
{I}_q(0)=4\eval{\partial^2_{\theta_1\theta_2}\ln|\tr \varrho_{\vb*{\theta}}(\tau)|}_{\vb*{\theta}=\vb*{0}},
\end{equation}
where $\varrho_{\vb*{\theta}}(\tau)=e^{\mca{L}_{\vb*{\theta}}\tau}(\pi)$ is an operator evolved according to the following modified super-operator:
\begin{align}
\mca{L}_{\vb*{\theta}}(\varrho)&=-i[H,\varrho]+\sum_{k\ge 1}\sqrt{e^{\omega_k(\epsilon[1+z_k\theta_1])+\omega_k(\epsilon[1+z_k\theta_2])}}V_k\varrho V_k^\dagger-\frac{1}{2}\sum_{k\ge 1}e^{\omega_k(\epsilon[1+z_k\theta_1])}V_k^\dagger V_k\varrho-\frac{1}{2}\sum_{k\ge 1}e^{\omega_k(\epsilon[1+z_k\theta_2])}\varrho V_k^\dagger V_k.
\end{align}
Following the same procedure outlined in Appendix \ref{app:proof.qtkur}, we obtain the explicit form of the quantum Fisher information as
\begin{equation}
{I}_q(0)=\tau\epsilon^2\sum_{k\ge 1}\tr(L_k\pi L_k^\dagger)=\tau\epsilon^2a,\label{eq:qFI.pert}
\end{equation}
which is proportional to the dynamical activity.
On the other hand, the partial derivative of the observable average with respect to $\theta$ can be calculated as
\begin{align}
\eval{\partial_{\theta}\expval{f(\vb*{\phi})}}_{\theta=0}&=\eval{\sum_{k\ge 1}\partial_\theta \omega_k(\epsilon[1+z_k\theta])d_{\omega_k(\epsilon[1+z_k\theta])}\ev{f(\vb*{\phi})}}_{\theta=0}\notag\\
&=\epsilon\sum_{k\ge 1}z_kd_\epsilon\omega_k(\epsilon)d_{\omega_k(\epsilon)}\ev{f(\vb*{\phi})}\notag\\
&=\epsilon\sum_{k\ge 1}|d_{\omega_k}\ev{f(\vb*{\phi})}|.\label{eq:p.avg.pert}
\end{align}
Substituting Eqs.~\eqref{eq:qFI.pert} and \eqref{eq:p.avg.pert} into Eq.~\eqref{eq:q.CR.bound.p} yields the desired relation \eqref{eq:main.result.3}.

\subsection{Generalization of relation \eqref{eq:main.result.3} to the transient regime}\label{app:res.kur.tran.gen}
Here we derive the generalization of relation \eqref{eq:main.result.3} to transient processes, where the system starts from an arbitrary quantum state $\varrho_0$. 
To achieve this, we consider an auxiliary dynamics perturbed similarly to Eq.~\eqref{eq:res.per.dyn}. By following the same calculations as in Eq.~\eqref{eq:qFI.tmp1}, the quantum Fisher information can be computed as
\begin{align}
	{I}_q(0)&=-4\eval{\bvec{\mbb{1}}\int_0^\tau\dd{t}e^{\widehat{\mca{L}}_{\vb*{\theta}}(\tau-t)}\partial_{\theta_1}\widehat{\mca{L}}_{\vb*{\theta}}e^{\widehat{\mca{L}}_{\vb*{\theta}}t}\kvec{\varrho_0}\bvec{\mbb{1}}\int_0^\tau\dd{t}e^{\widehat{\mca{L}}_{\vb*{\theta}}(\tau-t)}\partial_{\theta_2}\widehat{\mca{L}}_{\vb*{\theta}}e^{\widehat{\mca{L}}_{\vb*{\theta}}t}\kvec{\varrho_0}}_{\vb*{\theta}=\vb*{0}}\notag\\
	&+4\eval{\bvec{\mbb{1}}\int_0^\tau\dd{t}\int_0^{\tau-t}\dd{s}e^{\widehat{\mca{L}}_{\vb*{\theta}}(\tau-t-s)}\partial_{\theta_2}\widehat{\mca{L}}_{\vb*{\theta}}e^{\widehat{\mca{L}}_{\vb*{\theta}}s}\partial_{\theta_1}\widehat{\mca{L}}_{\vb*{\theta}}e^{\widehat{\mca{L}}_{\vb*{\theta}}t}\kvec{\varrho_0}}_{\vb*{\theta}=\vb*{0}}\notag\\
	&+4\eval{\bvec{\mbb{1}}\int_0^\tau\dd{t}\int_0^{t}\dd{s}e^{\widehat{\mca{L}}_{\vb*{\theta}}(\tau-t)}\partial_{\theta_1}\widehat{\mca{L}}_{\vb*{\theta}}e^{\widehat{\mca{L}}_{\vb*{\theta}}(t-s)}\partial_{\theta_2}\widehat{\mca{L}}_{\vb*{\theta}}e^{\widehat{\mca{L}}_{\vb*{\theta}}s}\kvec{\varrho_0}}_{\vb*{\theta}=\vb*{0}}\notag\\
	&+4\eval{\bvec{\mbb{1}}\int_0^\tau\dd{t}e^{\widehat{\mca{L}}_{\vb*{\theta}}(\tau-t)}\partial_{\theta_1\theta_2}^2\widehat{\mca{L}}_{\vb*{\theta}}e^{\widehat{\mca{L}}_{\vb*{\theta}}t}\kvec{\varrho_0}}_{\vb*{\theta}=\vb*{0}}\notag\\
	&=4\Big[-\int_0^\tau\dd{t}\bvec{\mbb{1}}\widehat{\mca{K}}_1\kvec{\varrho_t}\int_0^\tau\dd{t}\bvec{\mbb{1}}\widehat{\mca{K}}_2\kvec{\varrho_t}+\int_0^\tau\dd{t}\int_0^{\tau-t}\dd{s}\bvec{\mbb{1}}\widehat{\mca{K}}_2e^{\widehat{\mca{L}}s}\widehat{\mca{K}}_1\kvec{\varrho_t}\notag\\
	&+\int_0^\tau\dd{t}\int_0^{t}\dd{s}\bvec{\mbb{1}}\widehat{\mca{K}}_1e^{\widehat{\mca{L}}(t-s)}\widehat{\mca{K}}_2\kvec{\varrho_s}+\int_0^\tau\dd{t}\bvec{\mbb{1}}\eval{\partial_{\theta_1\theta_2}^2\widehat{\mca{L}}_{\vb*{\theta}}\kvec{\varrho_t}}_{\vb*{\theta}=\vb*{0}}\Big],
\end{align}
where the operators $\widehat{\mca{K}}_1$ and $\widehat{\mca{K}}_2$ are given by
\begin{align}
	\widehat{\mca{K}}_1&\coloneqq\eval{\partial_{\theta_1}\widehat{\mca{L}}_{\vb*{\theta}}}_{\vb*{\theta}=\vb*{0}}=\frac{1}{2}\sum_{k\ge 1}\epsilon z_kd_\epsilon\omega_k(\epsilon)\qty[ L_k\otimes L_k^* - (L_k^\dagger L_k )\otimes\mbb{1}],\\
	\widehat{\mca{K}}_2&\coloneqq\eval{\partial_{\theta_2}\widehat{\mca{L}}_{\vb*{\theta}}}_{\vb*{\theta}=\vb*{0}}=\frac{1}{2}\sum_{k\ge 1}\epsilon z_kd_\epsilon\omega_k(\epsilon)\qty[ L_k\otimes L_k^* - \mbb{1}\otimes(L_k^\dagger L_k )^\top].
\end{align}
Noting that $\bvec{\mbb{1}}\widehat{\mca{K}}_1\kvec{A}=\bvec{\mbb{1}}\widehat{\mca{K}}_2\kvec{A}=0$ for any operator $A$ and $|z_kd_\epsilon\omega_k(\epsilon)|=|\sign[d_{\omega_k(\epsilon)}\ev{f(\vb*{\phi})}]|=1$, we obtain
\begin{align}
	I_q(0)&=4\int_0^\tau\dd{t}\bvec{\mbb{1}}\eval{\partial_{\theta_1\theta_2}^2\widehat{\mca{L}}_{\vb*{\theta}}\kvec{\varrho_t}}_{\vb*{\theta}=\vb*{0}}\notag\\
	&=\int_0^\tau\dd{t}\bvec{\mbb{1}}\sum_{k\ge 1}[\epsilon z_kd_\epsilon\omega_k(\epsilon)]^2(L_k\otimes L_k^*)\kvec{\varrho_t}\notag\\
	&=\epsilon^2\int_0^\tau\dd{t}\sum_{k\ge 1}\tr(L_k\varrho_tL_k^\dagger)\notag\\
	&=\epsilon^2\mca{A}_\tau.
\end{align}
On the other hand, the partial derivative of the observable average with respect to $\theta$ retains the same form as in the steady-state case, $\eval{\partial_{\theta}\expval{f(\vb*{\phi})}}_{\theta=0}=\epsilon\sum_{k\ge 1}|d_{\omega_k}\ev{f(\vb*{\phi})}|$.
Consequently, we obtain the generalization of Eq.~\eqref{eq:main.result.3} as
\begin{equation}
	\frac{\|\nabla\ev{f(\vb*{\phi})}\|_1^2}{\mvar[f(\vb*{\phi})]}\le\mca{A}_\tau.
\end{equation}

\subsection{Recovery of the kinetic uncertainty relation in the classical limit}\label{app:ckur.recover}
We show that relation \eqref{eq:main.result.3} derives the classical KUR for Markov jump processes.
In the classical limit, the GKSL dynamics is equivalent to Markov jump dynamics, characterized by the transition matrix ${W}=[w_{mn}]\in\mbb{R}^{d\times d}$.
Let $w_{mn}=e^{\omega_{mn}}$ for each $m\neq n$.
Then, by applying the triangle inequality, we obtain
\begin{equation}
	\|\nabla\ev{\phi}\|_1=\sum_{m\neq n}|d_{\omega_{mn}}\ev{\phi}|\ge\qty|\sum_{m\neq n}d_{\omega_{mn}}\ev{\phi}|.
\end{equation}
In what follows, we prove that $\sum_{m\neq n}d_{\omega_{mn}}\ev{\phi}=\ev{\phi}$, which immediately yields the KUR.
Let $\{c_{mn}\}$ be the counting coefficients, then $\ev{\phi}=\tau\sum_{m\neq n}c_{mn}w_{mn}\pi_n$.
Using this, we can calculate as follows:
\begin{align}
	\sum_{m\neq n}d_{\omega_{mn}}\ev{\phi}&=\tau\sum_{i\neq j}c_{ij}\sum_{m\neq n}d_{\omega_{mn}}(w_{ij}\pi_j)\notag\\
	&=\tau\sum_{i\neq j}c_{ij}\sum_{m\neq n}[(d_{\omega_{mn}}w_{ij})\pi_j+w_{ij}d_{\omega_{mn}}\pi_j]\notag\\
	&=\tau\sum_{i\neq j}c_{ij}\sum_{m\neq n}(\delta_{mi}\delta_{nj}w_{ij}\pi_j+w_{ij}d_{\omega_{mn}}\pi_j)\notag\\
	&=\tau\sum_{i\neq j}c_{ij}w_{ij}\pi_j+\tau\sum_{i\neq j}c_{ij}\sum_{m\neq n}w_{ij}d_{\omega_{mn}}\pi_j\notag\\
	&=\ev{\phi}+\tau\sum_{i\neq j}c_{ij}w_{ij}\sum_{m\neq n}d_{\omega_{mn}}\pi_j.\label{eq:rec.kur.tmp1}
\end{align}
From the equalities ${W}\vpi=\vb*{0}$ and $\vb*{1}^\top\vpi=\vb*{1}$, we get ${Z}\vpi=\vb*{e}_1$, where $e_{nk}=\delta_{nk}$ and ${Z}$ is obtained from ${W}$ by replacing the first row of ${W}$ with $\vb*{1}$.
Here, $\vb*{1}=[1,\dots,1]^\top$ is the all-one vector.
Note that matrix ${Z}$ is invertible \cite{Aslyamov.2024.PRL}.
Taking the derivative of equality ${Z}\vpi=\vb*{e}_1$ with respect to $\omega_{mn}$, we get
\begin{align}
	(d_{\omega_{mn}}{Z})\vpi+{Z}d_{\omega_{mn}}\vpi=0\rightarrow d_{\omega_{mn}}\pi_j=-\vb*{e}_j^\top{Z}^{-1}(d_{\omega_{mn}}{Z})\vpi.
\end{align}
Noting that $\sum_{m\neq n}d_{\omega_{mn}}{Z}={Z}-\vb*{e}_1\vb*{1}^\top$, we obtain
\begin{align}
	\sum_{m\neq n}d_{\omega_{mn}}\pi_j&=-\vb*{e}_j^\top{Z}^{-1}(\sum_{m\neq n}d_{\omega_{mn}}{Z})\vpi\notag\\
	&=-\vb*{e}_j^\top{Z}^{-1}({Z}-\vb*{e}_1\vb*{1}^\top)\vpi\notag\\
	&=-\pi_j+[{Z}^{-1}]_{j1}.\label{eq:rec.kur.tmp2}
\end{align}
From ${Z}\vpi=\vb*{e}_1$, we have $\vpi={Z}^{-1}\vb*{e}_1$, which derives $\pi_j=[{Z}^{-1}]_{j1}$.
Consequently, $\sum_{m\neq n}d_{\omega_{mn}}\pi_j=0$ is obtained from Eq.~\eqref{eq:rec.kur.tmp2}.
Substituting this to Eq.~\eqref{eq:rec.kur.tmp1} yields $\sum_{m\neq n}d_{\omega_{mn}}\ev{\phi}=\ev{\phi}$, which recovers the KUR from inequality \eqref{eq:main.result.3}.

\twocolumngrid


\begin{thebibliography}{109}
\bibitem [{fnt()}]{fnt1}%
  \BibitemOpen
  \href@noop {} {}Simple algebraic calculations can show that the jump operators in the rotating frame is explicitly given by $\tilde{L}_k(t)=U_t^\dagger L_kU_t=e^{ix_kt}L_k$, where $x_k$ is a real number determined by $\{\omega_i\}_{1\le i\le 3}$. The phase factor $e^{ix_kt}$ does not affect the dissipative terms in the GKSL equation, thus ensuring that the dissipator remains structurally unchanged in the rotating frame.\BibitemShut {Stop}%
\end{thebibliography}
\end{document}